\documentclass[11pt]{article}

\usepackage{tikz} 
\usetikzlibrary{mindmap}
\usetikzlibrary{calc}
\usetikzlibrary{through}

\usepackage{bbm,amsfonts,amsmath,amssymb,amsthm}
\usepackage{fullpage}

\newtheorem{theorem}{Theorem}[section]
\newtheorem{definition}[theorem]{Definition}
\newtheorem{lemma}[theorem]{Lemma}
\newtheorem{corollary}[theorem]{Corollary}
\newtheorem{remark}[theorem]{Remark}
\newtheorem{example}[theorem]{Example}

\newcommand{\sm}{\setminus}
\newcommand{\eset}{\emptyset}
\newcommand{\seq}{\subseteq}
\newcommand{\ra}{\rightarrow}
\newcommand{\ua}{\uparrow}
\newcommand{\da}{\downarrow}

\newcommand{\cS}{{\mathcal S}}
\newcommand{\DP}{\mathrm{DP}}
\newcommand{\id}{\mathrm{id}}
\newcommand{\lepol}{\le_{pol}}
\newcommand{\sprod}[2]{\langle #1 , #2 \rangle}
\newcommand{\transition}[1]{\stackrel{#1}{\longrightarrow}}
\newcommand{\FRONT}{\mathrm{FRONTIER}} 
\newcommand{\CONS}{\mathrm{CONSISTENT}}
\newcommand{\FLAT}{\mathrm{FLATTEN}}

\title{The Word Problem for Products of Symmetric Groups} 
\author{Hans U. Simon\thanks{Faculty of Mathematics, Ruhr-University Bochum,
Germany. E-mail:{\tt hans.simon@rub.de}}}

\begin{document}

\maketitle

\begin{abstract}
The word problem for products of symmetric groups (WPPSG) is a well-known
NP-complete problem. An input instance of this problem consists
of ``specification sets'' $X_1,\ldots,X_m \seq \{1,\ldots,n\}$
and a permutation $\tau$ on $\{1,\ldots,n\}$. The sets $X_1,\ldots,X_m$
specify a subset of the symmetric group $\cS_n$ and the question
is whether the given permutation $\tau$ is a member of this subset.
We discuss three subproblems of WPPSG and show that they can be solved
efficiently. The subproblem WPPSG$_0$ is the restriction of WPPSG to
specification sets all of which are sets of consecutive integers.
The subproblem WPPSG$_1$ is the restriction of WPPSG to specification
sets which have the Consecutive Ones Property. The subproblem WPPSG$_2$
is the restriction of WPPSG to specification sets which have what we
call the Weak Consecutive Ones Property. WPPSG$_1$ is more general
than WPPSG$_0$ and WPPSG$_2$ is more general than WPPSG$_1$.
But the efficient algorithms that we use for solving WPPSG$_1$
and WPPSG$_2$ have, as a sub-routine, the efficient algorithm 
for solving WPPSG$_0$.
\end{abstract}

\section{Introduction}

As described in~\cite{GJMP1980,KLPS2007}, the ``Word Problem for Products
of Symmetric Groups (WPPSG)'' is intimately related to the problems
``Circular Arc Coloring (CAC)'' and ``Interval Scheduling with Machine 
Availabilities (ISMA)''. While the efficient solvability of subproblems
like WPPSG$_0$, WPPSG$_1$ and WPPSG$_2$ is likely to have implications on 
the corresponding subproblems of CAC and ISMA, our interest in this paper
is more in developing a general pattern according to which efficiently
solvable subproblems can be found. This pattern, formulated here for 
an abstract NP-hard decision problem $\DP$, is roughly as follows:
\begin{itemize}
\item
In order to get started, find a first basic subproblem $\DP_0$
of $\DP$ that can be solved efficiently.
\item
In order to get an efficient algorithm for a more general subproblem,
find an equivalence relation $R$ on the input instances which
satisfies the following constraints:
\begin{itemize}
\item
$R$ \emph{respects $\DP$} in the sense that each equivalence class 
contains either only YES-instances or only NO-instances. Call an input 
instance of $\DP$ ``nice'' if its equivalence class contains an instance 
of $\DP_0$.
\item
$R$ allows for an efficient algorithm that decides niceness.
\item
$R$ allows for an efficient algorithm which, given a nice instance $I$, 
returns an instance $I'$ of $\DP_0$ that falls into the same equivalence 
class as $I$.
\end{itemize}
\end{itemize}
Let $\DP_1$ denote the problem $\DP$ restricted to nice instances.
Since an instance of $\DP_0$ is clearly nice, $\DP_1$ is more general
than $\DP_0$. Moreover $\DP_1$ can be solved in the following manner:
\begin{enumerate}
\item
Decide whether the given input instance $I$ is nice. If this is not the case,
then reject $I$ and stop. Otherwise, if $I$ is nice, then proceed to Step 2.
\item
Find an instance $I'$ of $\DP_0$ that falls into the same equivalence class 
as $I$.
\item
Use the algorithm for the basic subproblem $\DP_0$ to decide whether $I'$
is a YES-instance (which is the case iff $I$ is a YES-instance).
\end{enumerate}

In this paper $\DP =$ WPPSG and we will present two equivalence relations
with the required properties. The first one is based on inner automorphisms
of $\cS_n$ and leads to an efficient algorithm for solving WPPSG restricted 
to input instances whose specification sets have the Consecutive Ones Property.
The second one is based on a more complicated transformation technique
and leads to an efficient algorithm for solving WPPSG restricted to input 
instances whose specification sets have the Weak Consecutive Ones Property. 
The algebraic flavor of WPPSG was very helpful for finding these
equivalence relations.

We suspect that the general pattern that we mentioned above can
find a wider applicability, in particular if the underlying decision 
problem has an algebraic flavor.

\section{Definitions and Notations} \label{sec:definitions}

For each integer $n \ge 1$, we set $[n] = \{1,\ldots,n\}$.
More generally, for integers $n_2 \ge n_1$, 
we set $[n_1:n_2] = \{n_1,n_1+1,\ldots,n_2\}$. 
For each integer $n\ge1$, 
let $\cS_n$ denote the symmetric group of all permutations on $[n]$. 
For $X \seq [n]$, let $\cS_X$ denote the subgroup of $\cS_n$ consisting 
of all permutations that leave the elements outside $X$ fixed. 
Let $\id$ denote the identity permutation.
Clearly $\cS_\eset = \{\id\}$ and $\cS_X = \{\id\}$ if $|X| = 1$. 
The composition of $m\ge2$ permutations $\sigma_1,\ldots,\sigma_m \in \cS_n$ 
will be simply denoted by $\sigma_1 \ldots \sigma_m$ with the 
understanding that $\sigma_1$ is applied first, $\sigma_2$ is applied second,
and so on. To emphasize that $\sigma_1$ is applied first, we will denote
the $\sigma_1 \ldots \sigma_m$-image of $i \in [n]$ by $i\sigma_1\ldots\sigma_m$.
Similarly the $\sigma_1 \ldots \sigma_m$-image of $X \seq [n]$ is
written as $X \sigma_1\ldots\sigma_m$, 
i.e., $X \sigma_1\ldots\sigma_m = \{i \sigma_1\ldots\sigma_m: i \in X\}$.
Permutations can be represented in different ways. 
For $1 \le i \neq j \le n$, we will denote by $\langle i,j \rangle$ 
the transposition which exchanges the elements $i$ and $j$.
We call $\langle i,j \rangle$ an \emph{adjacent transposition} 
if $i$ and $j$ are consecutive integers, i.e., if $j = i\pm1$.
In general, we will identify a permutation $\sigma$ with the 
tuple $(1\sigma,2\sigma,\ldots,n\sigma) \in [n]^n$. We say that $i_1\sigma$
\emph{precedes $i_2\sigma$ in $\sigma$} if $i_1 < i_2$. A sequence
\begin{equation} \label{eq:transpositions}
S = \langle i_1,j_1 \rangle , \langle i_2,j_2 \rangle, \ldots,
\langle i_s,j_s \rangle
\end{equation}
of (not necessarily different) transpositions on $[n]$ is called \emph{universal} 
if, for any permutation $\tau \in \cS_n$, there exists a subsequence $T$ 
of $S$ such that $\tau$ is the composition of the transpositions in $T$. 
If $S$ is a sequence of transpositions, then we denote by $\tilde S$ the
corresponding sequence of 2-sets, i.e., if $S$ is given by~(\ref{eq:transpositions}),
then
\[
\tilde S =  \{i_1,j_1\} , \{i_2,j_2\}, \ldots, \{i_s,j_s\} \enspace .
\] 
A universal sequence consisting of $\binom{n}{2}$ transpositions was
presented in~\cite{GJMP1980}. We will see now that there is even 
a universal sequence consisting of $\binom{n}{2}$ \emph{adjacent} 
transpositions:

\begin{lemma} \label{lem:universal-sequence}
The sequence $S = S',S''$ with
\begin{eqnarray*}
S' & = &
\langle 1,2 \rangle , \langle 2,3 \rangle ,\ldots, \langle n-1,n \rangle \\
S'' & = &
\langle 1,2 \rangle , \langle 2,3 \rangle ,\ldots, \langle n-2,n-1 \rangle\ ;\ 
\langle 1,2 \rangle , \langle 2,3 \rangle ,\ldots, \langle n-3,n-2 \rangle\ ;\ldots;  
\langle 1,2 \rangle , \langle 2,3 \rangle\ ;\  \langle 1,2 \rangle
\end{eqnarray*}
is a universal sequence consting of $\binom{n}{2}$ adjacent transpositions.
\end{lemma}

\begin{proof}
The sequence $S$ obviously consists of $(n-1)+(n-2)+\ldots+2+1 = \binom{n}{2}$
adjacent transpositions. It suffices therefore to show that $S$ is universal.
To this end, we consider  an arbitrary but fixed permutation $\tau\in\cS_n$.  
Set $i' := n\tau^{-1}$. If $i' = n$, then $\tau$ could be viewed as a
permutation on $[n-1]$, and we were done by induction. We may suppose
therefore that $i'<n$. Consider the permutations  
\[
i\tau' = \left\{ \begin{array}{cl}
            i & \mbox{if $i<i'$} \\
            n & \mbox{if $i=i'$} \\
            i-1 & \mbox{if $i>i'$}
           \end{array} \right. \mbox{ and }\ 
i\tau'' = \left\{ \begin{array}{cl}
            i\tau & \mbox{if $i<i'$} \\
            (i+1)\tau) & \mbox{if $i' \le i \le n-1$} \\
            n & \mbox{if $i=n$}
           \end{array} \right. \enspace .
\]
It is easily checked that $\tau = \tau'\tau''$.
Moreover $\tau' := 
\langle i',i'+1 \rangle \langle i'+1,i'+2 \rangle \ldots \langle n-1,n \rangle$.
In other words, $\tau'$ is the composition of the last $n-i'$ transpositions 
in the sequence $S'$. The latter form a subsequence of $S'$ that we denote 
by $T'$. Because of $\tau''(n) = n$, it follows inductively that there exists 
a subsequence $T''$ of $S''$ such that $\tau''$ is the composition of 
the transpositions in $T''$. From this discussion, it follows that $\tau$ 
can be written as the composition of all transpositions in the 
subsequence $T := T' , T''$ of $S$, which completes the proof.
\end{proof}

For sake of simplicity, we have introduced universal sequences of
transpositions (resp.~adjacent transpositions) on the interval $[n]$.
But these notions are defined in the obvious way for arbitrary
integer intervals $[n_1:n_2]$ with $n_1 \le n_2$. An interval of this 
form contains $n_2-n_1+1$ integers and admits a universal sequence 
of $\binom{n_2-n_1+1}{2}$ many adjacent transpositions.

\medskip
We remind the reader to the definition of matrices with the 
\emph{consecutive 1's property for columns (C1-property)}:

\begin{definition}[\cite{G1980}] \label{def:C1-property}
A matrix whose entries are zeros and ones is said to have the 
\emph{C1-property} if the rows can be permuted in such a way 
that the 1's in each column occur consecutively.
\end{definition}

When a column $A_j$ of a matrix $A \in \{0,1\}^{n \times m}$ is
identified with the set $X_j = \{i \in [n]: A[i,j] = 1\}$,
then Definition~\ref{def:C1-property} can be equivalently
reformulated as follows: \\
Subsets $X_1,\ldots,X_m$ of $[n]$ are said to have the C1-property, 
if there exists a permutation $\pi$ on $[n]$ such that, for $j=1,\ldots,m$, 
the elements of $X_j$ occur consecutively within $\pi$.

\smallskip\noindent
For ease of later reference, we make the following observation: 

\begin{remark} \label{rem:C1P}
The elements of a set $X \seq [n]$ occur consecutively within $\pi$ 
iff $X\pi^{-1}$ is a set of consecutive integers.
\end{remark}
For instance, the elements of $X = \{2,5,7,9\}$ occur consecutively
within $\pi = (6,3,8,5,2,9,7,4,1)$, namely in positions $4,5,6,7$.
The permutation $\pi^{-1}$ equals $(9,5,2,8,4,1,7,3,6)$
and satisfies $X\pi^{-1} = \{4,5,6,7\}$.

As shown by Booth and Lueker in~\cite{BL1976}, one can test
in time $O(mn)$ whether a given matrix $A \in \{0,1\}^{n \times m}$ 
(or a given collection of subsets of $[n]$) has the C1-property.
Moreover, the permutation on $[n]$ which witnesses this property 
can also be computed in time $O(mn)$.\footnote{Note that $O(mn)$ 
is a linear time bound because the given matrix $A$ is of size $mn$.}

\section{The Word Problem for Products of Symmetric Groups}
\label{sec:wppsg}

The word problem for products of symmetric groups (WPPSG) is a well-known
NP-complete problem~\cite{GJMP1980}. In Section~\ref{subsec:version1},
we formally state this problem and also remind the reader to some 
(mostly known) polynomial reductions among some of its subproblems. 
In Section~\ref{subsec:version2}, the problem WPPSG is reformulated 
in a way that allows for a natural definition of what we call the 
\emph{sorting strategy}. This strategy is then intensively discussed
in Section~\ref{subsec:sorting-strategy}. The main result,
Theorem~\ref{th:sorting-strategy}, states that this strategy efficiently 
solves WPPSG restricted to specification sets all of which are sets of 
consecutive integers, respectively.

\subsection{The Classical Formulation of the Word Problem} 
\label{subsec:version1}

The \emph{word problem for products of symmetric groups (WPPSG)}
is the following decision problem:
\begin{description}
\item[Input Instance:]
Integers $n,m\ge1$, specification sets $X_1,\ldots,X_m \seq [n]$ and
a permutation $\tau_0$ on~$[n]$.
\item[Question:]
Is $\tau_0$ a member of the product of the symmetric 
groups $\cS_{X_1},\ldots,\cS_{X_m}$, i.e., do there exist 
permutations $\sigma_1 \in \cS_{X_1},\ldots,\sigma_m \in \cS_{X_m}$
such that $\tau_0 = \sigma_1 \ldots \sigma_m$?
\end{description}
Since $\cS_X = \{\id\}$ if $|X|\le1$, we may always assume without loss 
of generality that $n \ge |X_j| \ge 2$ for $j=1,\ldots,m$.

We will also consider some subproblems of WPPSG which are obtained 
by imposing additional constraints on the sets $X_1,\ldots,X_m$:
\begin{itemize}
\item
The subproblem WPPSG$[2]$ results
from the constraints $|X_j| = 2$ for $j = 1,\ldots,m$.
\item
The subproblem WPPSG$_0$ results from the constraint that, 
for each $j \in [m]$, the set $X_j \seq [n]$ consists of 
consecutive integers.
\item
The subproblem WPPSG$_0[2]$ is then obtained by combining 
the constraints from the previous two subproblems.
\item
The subproblem WPPSG$_1$ results from the constraint 
that the sets $X_1,\ldots,X_m \seq [n]$ must have the 
C1-property, i.e., there exists a permutation $\pi \in \cS_n$
such that, for every $j \in [m]$, the elements of $X_j$ occur
consecutively in $\pi$ (or, equivalently, the set $X_j\pi^{-1}$
consists of consecutive integers).
\end{itemize}
Clearly, as witnessed by the identity permutation, WPPSG$_0$ is a 
subproblem of WPPSG$_1$. The following hardness results are known:

\begin{theorem}[\cite{GJMP1980}] \label{th:hardness}
\begin{enumerate}
\item WPPSG is NP-complete.
\item WPPSG $\lepol$WPPSG$[2]$.
\item WPPSG$[2]$ is NP-complete.
\end{enumerate}
\end{theorem}

The proof of the first assertion employs a polynomial reduction from 
``Directed Disjoint Connecting Paths'' to WPPSG. The polynomial reduction 
which establishes the second assertion makes use of universal sequences 
of transpositions (without explicitly calling them such).
See~\cite{GJMP1980} for details. The third assertion is immediate from the 
first two. 

The following result makes use of universal sequences of adjacent
transpositions in a straightforward manner. We present a short proof 
for sake of completeness.

\begin{lemma}[Common Knowledge] \label{lem:reduction}
$\mbox{WPPSG}_0 \lepol \mbox{WPPSG}_0[2]$. 
\end{lemma}

\begin{proof}
The following proof is illustrated in Example~\ref{ex:reduction} below.
Let $I = (n,m,X_1,\ldots,X_m,\tau_0)$ be an instance of WPPSG$_0$.
Then the sets $X_1,\ldots,X_m$ are integer intervals. Let $S_j$ be a 
universal sequence of adjacent transpositions for the interval $X_j$
and let $\tilde S_j$ be the corresponding sequence of $2$-sets.
Consider the instance $I'$ of WPPSG$_0[2]$ which is obtained
from $I$ by substituting $\tilde S_j$ \hbox{for $X_j$}. $I'$ can clearly 
be constructed from $I$ in poly-time. It suffices to show that $I$
is a YES-instance iff $I'$ is a YES-instance. \\
Suppose first that $I'$ is a YES-instance. Then $\tau_0$ can be 
written in the form $\tau_0 = \sigma_1 \ldots \sigma_m$ where $\sigma_j$
is a product of transpositions of a suitably chosen subsequence of $S_j$. 
Obviously $\sigma_j \in \cS_{X_j}$, which shows that $I$ is a YES-instance. \\
Suppose now that $I$ is a YES-instance. Then $\tau_0$ can be written
in the form $\tau_0 = \sigma_1 \ldots \sigma_m$ where $\sigma \in \cS_{X_j}$.
By the universality of the sequence $S_j$, each $\sigma_j$ can be written
as a product of transpositions of a suitably chosen subsequence of $S_j$. 
This shows that $I'$ is a YES-instance.
\end{proof}

\begin{example} \label{ex:reduction}
Consider the instance $I = (n,m,X_1,\ldots,X_m,\tau_0)$ of WPPSG$_0$
where $n=5$, $m=4$, $X_1 = \{2,3\}$, $X_2 = \{1,2\}$, $X_3 = \{3,4,5\}$
and $X_4 = \{1,2,3,4\}$. Let $I' = (n,m,\tilde S_1,\ldots, \tilde S_m,\tau_0)$ 
be the corresponding instance of WPPSG$_0[2]$, as described in the
proof of Lemma~\ref{lem:reduction}. Then
\[
I' = (5,4,\ \underbrace{\{2,3\}}_{\tilde S_1}\ ,\ \underbrace{\{1,2\}}_{\tilde S_2}\ ,\ 
\underbrace{\{3,4\},\{4,5\},\{3,4\}}_{\tilde S_3}\ ,\ 
\underbrace{\{1,2\},\{2,3\},\{3,4\},\{1,2\},\{2,3\},\{1,2\}}_{\tilde S_4}\ ,\ \tau_0) 
\]
For $\tau_0 = (4,5,3,2,1)$, $I$ is a YES-instance 
because $\tau_0 = \sigma_1\sigma_2\sigma_3\sigma_4$ for 
\[
\sigma_1 = \sprod{2}{3}\ ,\ \sigma_2 = \sprod{1}{2}\ ,\ 
\sigma_3 = \sprod{3}{5}\ ,\ \sigma_4 = \sprod{1}{3}\sprod{2}{4} \enspace .
\]
Since $\sprod{3}{5} = \sprod{3}{4}\sprod{4}{5}\sprod{3}{4}$
and $\sprod{1}{3}\sprod{2}{4} = \sprod{2}{3}\sprod{3}{4}\sprod{1}{2}\sprod{2}{3}$,
$I'$ is a YES-instance too. It is no coincidence that, for each $j \in [4]$,
the permutation $\sigma_j$ can be written as a product of adjacent transpositions
taken from $S_j$. Rather it follows from the universality of the sequence $S_j$.
\end{example}

\subsection{A Reformulation of the Word Problem} \label{subsec:version2}

Consider two permutations $\tau,\tau' \in \cS_n$ and 
a set $X \seq [n]$. We write $\tau' \transition{X} \tau$ 
if $i\tau = i\tau'$ holds for each $i \in [n] \sm X$. 
Note that and $\tau' \transition{X} \tau$
implies that $\{i\tau: i \in X\} = \{i\tau': i \in X\}$.
Moreover $\transition{X}$ is an equivalence relation.
We say that $\tau$ is \emph{$X$-sorted} if, for any two
indices $i < j$ from $X$ we have that $i\tau < j\tau$.
Clearly every equivalence class with respect to $\transition{X}$
contains exactly one representative that is $X$-sorted.
The following observation is rather obvious:

\begin{lemma} \label{lem:X-transition}
Let $\tau$ and $\tau'$ be two permutations on $[n]$,
let $\sigma = \tau' \tau^{-1}$ and let $X \seq [n]$.
Then $\tau' \transition{X} \tau$ iff $\sigma \in \cS_X$.
\end{lemma}

\begin{proof}
Suppose first that $\tau' \transition{X} \tau$
so that $\tau$ coincides with $\tau'$ outside $X$. 
Pick an arbitrary but fixed $i \in [n] \sm X$.
Then $i\sigma = i \tau' \tau^{-1} = i \tau \tau^{-1} = i$.
It follows that $\sigma \in \cS_X$. \\
Suppose now that $\sigma \in \cS_X$. 
Pick an arbitrary but fixed $i \in [n] \sm X$.
Then $i = i\sigma = i \tau' \tau^{-1}$, which implies 
that $i\tau = i\tau'$. 
It follows that $\tau' \transition{X} \tau$. 
\end{proof}

\noindent
We will see shortly that WPPSG can be reformulated in the following way:
\begin{description}
\item[Input Instance:]
Integers $n,m\ge1$, subsets $X_1,\ldots,X_m \seq [n]$ and
a permutation $\tau_0 \in \cS_n$.
\item[Question:]
Do there exist permutations $\tau_1,\ldots,\tau_m \in \cS_n$ 
such that $\tau_m = (1,\ldots,n)$ and, for $j = 1,\ldots,m$, 
we have that $\tau_{j-1} \transition{X_j} \tau_j$?
\end{description}
We will refer to this reformulated version of WPPSG as V2-WPPSG
where ``V2'' means ``Version 2''. The following result and its proof
show that the problems WPPSG and V2-WPPSG are equivalent in a strong 
sense:

\begin{lemma} \label{lem:V2-WPPSG}
WPPSG and V2-WPPSG have the same set of YES-instances.
\end{lemma}

\begin{proof}
Suppose first that $I = (n,m,X_1,\ldots,X_m,\tau_0)$ is a YES-instance 
of V2-WPPSG. Choose $\tau_1,\ldots,\tau_{m-1},\tau_m \in \cS_n$ such 
that $\tau_m = (1,\ldots,n)$ and $\tau_{j-1}\transition{X_j}\tau_j$ 
for $j = 1,\ldots,m$. For $j =1,\ldots,m$, 
set $\sigma_j := \tau_{j-1}\tau_j^{-1}$.
Then $\sigma_1\ldots\sigma_m = \tau_0\tau_m^{-1} = \tau_0$. Moreover,
according to Lemma~\ref{lem:X-transition}, we have 
that $\sigma_j \in \cS_{X_j}$ for $j=1,\ldots,m$. As witnessed 
by $\sigma_1,\ldots,\sigma_m$, $I$ is also a YES-instance of WPPSG. \\
Suppose now that $I = (n,m,X_1,\ldots,X_m,\tau_0)$ is a YES-instance 
of WPPSG. For $j = 1,\ldots,m$, choose $\sigma_j \in \cS_{X_j}$ such 
that $\tau_0 = \sigma_1\ldots,\sigma_m$ and choose $\tau_j$ such 
that $\sigma_j = \tau_{j-1}\tau_j^{-1}$.
Then $\tau_0 = \sigma_1\ldots\sigma_m = \tau_0\tau_m^{-1}$,
which implies that $\tau_m = \id = (1,\ldots,n)$. Moreover, 
according to Lemma~\ref{lem:X-transition}, we have 
that $\tau_{j-1}\transition{X_j}\tau_j$ for $j= 1,\ldots,m$. 
As witnessed by $\tau_1,\ldots,\tau_{m-1},\tau_m$, $I$ is also 
a YES-instance of V2-WPPSG.
\end{proof}

\begin{example} \label{ex:V2-WPPSG}
Consider the instance $I = (n,m,X_1,\ldots,X_m,\tau_0)$ of V2-WPPSG
where $n=5$, $m=4$, $X_1 = \{1,3\}$, $X_2 = \{4,5\}$, $X_3 = \{1,5\}$, 
$X_4 = \{2,3,4\}$ and $\tau_0 = (2,4,5,1,3)$. This is a YES-instance 
of V2-WPPSG because
\[
(2,4,5,1,3) \transition{1,3} (5,4,2,1,3) \transition{4,5} (5,4,2,3,1)
\transition{1,5} (1,4,2,3,5) \transition{2,3,4} (1,2,3,4,5) \enspace .
\]
According to Lemma~\ref{lem:V2-WPPSG}, the above instance $I$ must
be also a YES-instance of WPPSG. If we choose $\sigma_1,\ldots,\sigma_4$ 
as described in the proof of this lemma, we obtain
\[
\sigma_1 = (3,2,1,4,5)\ ,\ \sigma_2 = (1,2,3,5,4)\ ,\    
\sigma_3 = (5,2,3,4,1)\ ,\ \sigma_4 = (1,4,2,3,5)
\]
so that 
$\sigma_1 \in \cS_{\{1,3\}}$, $\sigma_2 \in \cS_{\{4,5\}}$,
$\sigma_3 \in \cS_{\{1,5\}}$, $\sigma_4 \in \cS_{\{2,3,4\}}$
and $\sigma_1\sigma_2\sigma_3\sigma_4 = \tau_0$.
\end{example}

\subsection{The Sorting Strategy} \label{subsec:sorting-strategy}

Let $I = (n,m,X_1,\ldots,X_m,\tau_0)$ be an instance of V2-WPPSG.
We say that the permutations $\tau_1,\ldots,\tau_m$ are chosen 
according to the \emph{sorting strategy (applied to $\tau_0$
and $X_1,\ldots,X_m$)} if, for $j=1,\ldots,m$, 
we choose $\tau_j$ as the unique permutation which is $X_j$-sorted 
and which coincides with $\tau_{j-1}$ outside $X_j$.\footnote{Often
$\tau_0$ and $X_1,\ldots,X_m$ are clear from context and need not
be mentioned explicitly.} Assume now that $I$ is a YES-instance 
of V2-WPPSG. We say that the sorting 
strategy is \emph{successful on $I$} if the final permutation $\tau_m$ 
chosen by this strategy equals the identity permutation $(1,\ldots,n)$. 
If $\tau_m \neq (1,\ldots,n)$, then we say that the sorting strategy 
\emph{fails on $I$}. The following example shows that the sorting 
strategy may fail on YES-instances of V2-WPPSG:

\begin{example} \label{ex:sorting-strategy}
We consider again the YES-instance from Example~\ref{ex:V2-WPPSG}.
The sorting strategy fails on this instance because it leads to
the following sequence of permutations:
\[
(2,4,5,1,3) \transition{1,3} (2,4,5,1,3) \transition{4,5} (2,4,5,1,3)
\transition{1,5} (2,4,5,1,3) \transition{2,3,4} (2,1,4,5,3) \enspace .
\]
\end{example}

\noindent
As for the subproblem V2-WPPSG$_0[2]$ however, we get the following
result:

\begin{theorem} \label{th:sorting-strategy}
The sorting strategy is successful on all YES-instances of V2-WPPSG$_0[2]$.
\end{theorem}

\begin{proof}
We prove the theorem by induction on $m$. The sorting strategy is clearly
successful on YES-instances with $m=1$. Suppose therefore
$(n,m,X_1,\ldots,X_m,\tau_0)$ is a YES-instance of V2-WPPSG$_0[2]$
with $m \ge 2$. Since each set $X_j$ must must contain two consecutive 
numbers in $[n]$, there exists $k_j \in \{1,\ldots,n-1\}$ such
that $X_j = \{k_j,k_j+1\}$. Because we are talking about a YES-instance,
there must exist permutations $\tau_1,\tau_2,\ldots,\tau_m$ such that
$\tau_0 \transition{X_1} \tau_1$ and
\begin{equation} \label{eq:chain1}
\tau_1 \transition{X_2} \tau_2 \transition{X_3}\ldots\transition{X_m} 
\tau_m = (1, \ldots, n) \enspace .
\end{equation}
By the inductive hypothesis, we may assume that, given $\tau_1$,
the permutations $\tau_2,\ldots,\tau_m$ are chosen according to 
the sorting strategy. If $\tau_1$ were $X_1$-sorted, then the 
theorem would follow directly. Suppose therefore that $\tau_1$ 
is not $X_1$-sorted, i.e., $k_1\tau_1 > (k_1+1)\tau_1$.
For sake of a simple notation, we set $q := k_1\tau_1$ 
and $p := (k_1+1)\tau_1$. Note that $q > p$, $q$ precedes $p$ 
in $\tau_1$ and $p$ precedes $q$ in $\tau_m$. Hence there must 
exist a $t \in [m-1]$ such that $\tau_t$ is of the
form $(v_t,q,p,w_t)$ and $\tau_{t+1}$ is of the form $(v_t,p,q,w_t)$. 
Since (\ref{eq:chain1}) is in accordance with the sorting strategy, 
it must be the case that $q$ precedes $p$ in $\tau_1,\ldots,\tau_t$ 
and $p$ precedes $q$ in $\tau_{t+1},\ldots,\tau_m$.
To complete the proof, it suffices to find a permutation sequence
\begin{equation} \label{eq:new-permutations}
\tau_0 \transition{X_1} \tau'_1 \transition{X_2} \tau'_2
\transition{X_3}\ldots\transition{X_m} \tau'_m = (1, \ldots, n) 
\end{equation}
with the property that $\tau'_1$ is $X_1$-sorted.\footnote{We would like to
emphasize that $\tau'_2,\ldots,\tau'_m$ do not have to be selected
according to the sorting strategy.} We define this sequence
in several stages:
\begin{description}
\item[Stage 1:]
Choose $\tau'_1$ as the unique $X_1$-sorted permutation that coincides 
with $\tau_0$ outside $X_1$. Note that $\tau'_1$ equals $\tau_1$ up to 
an exchange of $p$ and $q$ in the components $k_1$ and $k_1+1$. 
\item[Stage 2:]
Choose the permutations $\tau'_2,\ldots,\tau'_t$ so as to imitate 
the choices of $\tau_2,\ldots,\tau_t$, i.e., for $j=2,\ldots,t$, 
do the following:
\begin{itemize} 
\item
If $\tau_j = \tau_{j-1}$, then set $\tau'_j = \tau'_{j-1}$.
\item
If $\tau_j$ is obtained from $\tau_{j-1}$ by an exchange of the 
components $k_j$ and $k_j+1$, then let $\tau'_j$ be obtained 
from $\tau'_{j-1}$ by the same component exchange. 
\end{itemize}
Note that, in both cases, we have that $\tau'_{j-1} \transition{X_j} \tau'_j$.
\item[Stage 3:]
For $j = t+1,\ldots,m$, set $\tau'_j = \tau_j$.
\end{description}
It is easy to see that, for $j=1,\ldots,t$, the permutation $\tau'_j$ 
can be obtained from $\tau_j$ by an exchange of $p$ and $q$. 
From $\tau_t = (v_t,q,p,w_t)$ and $\tau_{t+1} = (v_t,p,q,w_t)$, 
we can infer that $\tau'_t = (v_t,p,q,w_t) = \tau_{t+1} = \tau'_{t+1}$. 
Thus the permutations $\tau'_j$ and $\tau_j$ are exactly the same 
for $j = t+1,\ldots,m$. In particular $\tau'_m = \tau_m = (1,\ldots,m)$.
We arrived at~(\ref{eq:new-permutations}) with $\tau'_1$ being $X_1$-sorted.
Now the theorem follows by induction.
\end{proof}

\noindent
Here are some immediate implications of Theorem~\ref{th:sorting-strategy}:

\begin{corollary} \label{cor:sorting}
\begin{enumerate}
\item
Let $S$ be a universal sequence of adjacent transpositions on $[n]$
and let $m = |S|$. Then, for each permutation $\tau$ on $[n]$, we 
have that  $I = (n,m,\tilde S,\tau)$ is a YES-instance of V2-WPPSG$_0[2]$. 
\item
Let $X = [a+1:a+n']$ be an integer sub-interval of $[n]$, let $S$ be 
a universal sequence of adjacent transpositions on $X$ and let $m = |S|$. 
Let $\tau_0$ be a permutation on $[n]$ and let $\tau_1,\ldots,\tau_m$
be the permutation sequence chosen by the sorting strategy when the latter
is applied to $\tau_0$ and the $m$ 2-sets in $\tilde S$. Then $\tau_m$ is 
the unique $X$-sorted permutation which coincides 
with $\tau_0$ \hbox{outside $X$}.
\end{enumerate}
\end{corollary}

\begin{proof}
\begin{enumerate}
\item
The universality of $S$ directly implies that $I$ is a YES-instance 
of V2-WPPSG$_0[2]$.
on $I$.
\item
For reasons of symmetry, we may assume that $a=0$ so that $X = [n']$.
Clearly $\tau_m$ coincides with $\tau_0$ outside $X$. We have to
show that $\tau_m(1) <\ldots< \tau_m(n')$. While this is easy to show,
the formal argument is a bit technical. See Example~\ref{ex:sorting}
below for an illustration. Let $\tau'_0$ be the permutation on $[n']$ 
that is obtained from $\tau_0$ as follows:
for $i=1,\ldots,n'$, set $i\tau'_0 = i'$ if $i\tau_0$ is the 
$i'$-smallest element in $X\tau_0$. Then $i_1\tau'_0 < i_2\tau'_0$ 
iff $i_1\tau_0 < i_2\tau_0$. Loosely speaking, the sorting 
strategy does not realize the substitution of $\tau'_0$ for $\tau_0$. 
Let $\tau'_1,\ldots,\tau'_m$ be the permutation sequence chosen by 
the sorting strategy when the latter is applied to $\tau'_0$ and 
the $m$ 2-sets in $\tilde S$. It now suffices to show that $\tau'_m$
is the identity permutation on $[n']$. But this follows from 
Theorem~\ref{th:sorting-strategy} because, according to the first assertion 
in the corollary, $(n',m,\tilde S,\tau'_0)$ is a YES-instance of WPPSG$_0[2]$,
\end{enumerate}
\end{proof}

\begin{example} \label{ex:sorting}
Let $n = 5$, $X = \{1,2,3\}$, $S = \sprod{1}{2} , \sprod{2}{3} , \sprod{1}{2}$
and $\tau_0 = (3,5,1,4,2)$. Then $\tau'_0 = (2,3,1)$ because $X\tau_0 = \{3,5,1\}$
with $1$ as the smallest member of $X\tau_0$, $3$ as the second- and $5$ as the
third-smallest. 
Let $\tau'_1,\tau'_2,\tau'_3$ be the permutation sequence chosen by 
the sorting strategy when the latter is applied to $\tau'_0$ and $\tilde S$. 
As argued in the proof of Corollary~\ref{cor:sorting}, the permutation $\tau'_3$ 
should be $X$-sorted. This is the case, indeed:
\[
\tau'_0 = (2,3,1) \transition{1,2} \tau'_1 = (2,3,1) \transition{2,3} 
\tau'_2 = (2,1,3) \transition{1,2} \tau'_3 = (1,2,3) \enspace .
\]
Note that the sorting strategy makes the analogous tuple-manipulations
when applied to $\tau_0$ and $\tilde S$:
\[
\tau_0 = (3,5,1,4,2) \transition{1,2} \tau'_1 = (3,5,1,4,2) \transition{2,3} 
\tau'_2 = (3,1,5,4,2) \transition{1,2} \tau'_3 = (1,3,5,4,2) \enspace .
\]
Specifically, it transforms $(3,5,1)$ into the sorted sequence $(1,3,5)$
as it transformed before $(2,3,1)$ into the identity permutation $(1,2,3)$.
\end{example}

\noindent
Combining Theorem~\ref{th:sorting-strategy} with Lemma~\ref{lem:reduction}
and Lemma~\ref{lem:V2-WPPSG}, we obtain the following result:

\begin{corollary}
The problems V2-WPPSG$_0$ and WPPSG$_0$ can be solved efficiently.
\end{corollary}

Can we solve V2-WPPSG$_0$ efficiently by applying the sorting strategy
directy? We give an affirmative answer to this question:

\begin{theorem} \label{th:WPPSG0}
The sorting strategy is successful on all YES-instances of V2-WPPSG$_0$.
\end{theorem}

\begin{proof}
Consider a YES-instance $I = (n,m,X_1,\ldots,X_m,\tau_0)$ of V2-WPPSG$_0$.
Then the sets $X_1,\ldots,$ $X_m$ are integer intervals. For $j = 1,\ldots,m$,
let $S_j$ and $\tilde S_j$ be defined as in the proof of 
Lemma~\ref{lem:reduction}. As in this proof, let $I'$ be the instance
obtained from $I$ by substituting $\tilde S_j$ for $X_j$. We know from 
Lemma~\ref{lem:reduction} that $I'$ is a YES-instance 
of V2-WPPSG$_0[2]$. Let
\[ 
\tau_0 \transition{X_1} \tau_1 \transition{X_2} \ldots \transition{X_m} \tau_m 
\]
be the permutation sequence chosen by the sorting strategy when
the latter is applied to $\tau_0$ and $X_1,\ldots,X_m$. It follows 
that $\tau_j$ is the unique $X_j$-sorted permutation that outside $X_j$ 
coincides with $\tau_{j-1}$. Let $\tau'_0,\tau'_1,\ldots,\tau'_m$ be 
a permutation sequence that is defined inductively as follows:
\begin{enumerate}
\item $\tau'_0 = \tau_0$.
\item 
For $j = 1,\ldots,m$, let $\tau'_j$ be the final permutation that we get 
after the sorting strategy has been applied to $\tau'_0$ and the 2-sets
$\tilde S_1,\ldots,\tilde S_j$.
\end{enumerate}
\begin{description}
\item[Claim:]
For $j = 0,1,\ldots,m$, we have that $\tau'_j = \tau_j$. 
\item[Proof of the Claim:]
The claim is clearly true for $j=0$. Assume inductively 
that $\tau'_{j-1} = \tau_{j-1}$. As mentioned above already,
the permutation $\tau_j$ is the unique $X_j$ sorted permutation
which coincides with $\tau_{j-1}$ outside $X_j$. Applying 
Corollary~\ref{cor:sorting} for $X_0 := X_j$, $S := S_j$ 
and $\tau_0 := \tau'_{j-1}$, we see that $\tau'_j$ is the unique 
$X_j$-sorted permutation which coincides with $\tau'_{j-1}$ 
outside $X_j$. Because $\tau'_{j-1} = \tau_{j-1}$, we obtain now 
that $\tau'_j = \tau_j$ which concludes the proof of the claim.
\end{description}
Because of the identity $\tau'_m = \tau_m$, the proof of the theorem is 
easy to accomplish. Since, as already mentioned above, $I'$ is a
YES-instance of WPPSG$_0[2]$, we may apply Theorem~\ref{th:sorting-strategy}
and conclude that $\tau'_m = (1,\ldots,m)$. Because of $\tau'_m = \tau_m$,
we have that $\tau_m = (1,\ldots,m)$. But this means that the sorting strategy
is successful on $I$.
\end{proof}

\section{Exploitation of the C1-Property}

We have seen in Section~\ref{subsec:sorting-strategy} that the problem
WPPSG$_0$ can be solved efficiently by means of the sorting strategy.
In this section, we show that even the more general problem WPPSG$_1$
can be solved efficiently. Recall that WPPSG$_1$, by definition, is the 
restriction of WPPSG to specification sets which have the C1-property. 
This will enable us to obtain an algorithm that solves WPPSG$_1$ efficiently 
by combining a well known algorithm of Booth and Lueker~\cite{BL1976} 
with the sorting strategy. 

The proofs of Lemmas~\ref{lem1:automorphism} and~\ref{lem2:automorphism}
below exploit that, for every $\pi \in \cS_n$, 
the mapping $\sigma \mapsto \pi^{-1}\sigma\pi$ is an inner automorphism 
of the group $\cS_n$. The proofs are straightforward but given here for 
the sake of completeness. 

\begin{lemma}[Common Knowledge] \label{lem1:automorphism}
Let $X$ be a subset of $[n]$ and let $\sigma , \pi$ be permutations
on $[n]$. Then $\sigma \in \cS_X$ 
iff $\pi^{-1}\sigma\pi \in \cS_{X\pi}$. 
\end{lemma}

\begin{proof}
Suppose first that $\sigma \in \cS_X$.
Choose an arbitrary $i \in [n] \sm (X\pi)$.
Then $i\pi^{-1} \notin X$ so that $i\pi^{-1}\sigma = i\pi^{-1}$.
Hence $i\pi^{-1}\sigma\pi = i\pi^{-1}\pi = i$.
This shows that $\pi^{-1}\sigma\pi \in \cS_{X\pi}$. \\ 
Suppose now that $\sigma' := \pi^{-1}\sigma\pi \in \cS_{X'}$ 
for $X' := X\pi$. We conclude from the first part of the proof
that $\pi\sigma'\pi^{-1} \in \cS_{X'\pi^{-1}}$. Since 
$\pi\sigma'\pi^{-1} = \pi(\pi^{-1}\sigma\pi)\pi^{-1} = \sigma$
and $X'\pi^{-1} = X\pi\pi^{-1} = X$, it follows 
that $\sigma \in \cS_X$.
\end{proof}

Let $I = (n,m,X_1,\ldots,X_m,\tau)$ be an instance of WPPSG and let $\pi$ 
be a permutation \hbox{on $[n]$}. We set
\begin{equation} \label{eq:equal-up-to}
I^\pi := (n,m,X_1\pi,\ldots,X_m\pi , \pi^{-1} \tau \pi) 
\enspace .
\end{equation}
We say that \emph{$I'$ equals $I$ up to renumbering} if $I' = I^\pi$
for some permutation $\pi$ on $[n]$. It is easily verified 
that $I^{\id} = I$ and $(I^{\pi_1})^{\pi_2} = I^{\pi_1\pi_2}$.
It follows that $(I^\pi)^{\pi^{-1}} = I^{\id} = I$.
It follows furthermore that being equal up to renumbering is an equivalence
relation so that it partitions the instances of WPPSG into equivalence classes.
The following result is quite obvious:

\begin{lemma}[Common Knowledge] \label{lem2:automorphism}
If $I = (n,m,X_1,\ldots,X_m,\tau)$ is a YES-instance of WPPSG, 
then every instance that equals $I$ up to renumbering is a 
YES-instance too.
\end{lemma}

\begin{proof}
Suppose first that $I$ is a YES-instance of WPPSG. Then there 
exist $\sigma_1\in\cS_{X_1},\ldots,\sigma_m \in\cS_{X_m}$ 
such that $\tau = \sigma_1 \ldots \sigma_m$. Pick now an
arbitrary but fixed permutation $\pi \in \cS_n$.
We have to show that $I^\pi = (n,m,X_1\pi,\ldots,X_m\pi,\pi^{-1}\tau\pi)$ is a YES-instance of WPPSG. Note first that
\[
\pi^{-1} \tau \pi =
\pi^{-1} \sigma_1 \sigma_2 \ldots \sigma_m \pi =
(\pi^{-1} \sigma_1 \pi) (\pi^{-1} \sigma_2 \pi) 
\ldots (\pi^{-1} \sigma_m \pi) \enspace .
\]
Moreover, according to Lemma~\ref{lem1:automorphism}, 
we have that $\pi^{-1} \sigma_j \pi \in \cS_{X_j\pi}$. 
Thus $I^\pi$ is a YES-instance.  
\end{proof}

\begin{description}
\item[Observation:]
Let $I = (n,m,X_1,\ldots,X_m,\tau_0)$ be an instance of WPPSG.
There exists an instance $I'$ of WPPSG$_0$ that equals $I$
up to renumbering iff the sets $X_1,\ldots,X_m \seq [n]$ have 
the C1-property. More precisely, for each permutation $\pi$ 
on $[n]$, the following three assertions are equivalent:
\begin{enumerate}
\item
$I^\pi$ is an instance of WPPSG$_0$ 
\item 
$\pi$ is a permutation on $[n]$ such that, for $j=1,\ldots,m$,
the set $X_j\pi \seq [n]$ consists of consecutive integers.
\item
$\pi$ is a permutation on $[n]$ such that, for $j=1,\ldots,m$,
the elements of $X_j$ occur consecutively within $\pi^{-1}$.
\end{enumerate}
\end{description}
The equivalence between the first two assertions follows directly
from the definitions of $I^\pi$ and WPPSG$_0$. The equivalence 
of the last two assertions is immediate from Remark~\ref{rem:C1P}.
We are now in the position to prove the following result:

\begin{theorem}
The problem WPPSG$_1$ can be solved efficiently.
\end{theorem}

\begin{proof}
We use the method of Booth and Lueker~\cite{BL1976}, to test
whether the sets $X_1,\ldots,X_m$ have the C1-property
and, if applicable, to get a permutation $\pi$ such 
that $I^\pi$ is an instance of WPPSG$_0$. If $X_1,\ldots,X_m$
do not have the C1-property, then $I$ is rejected. Otherwise,
if $X_1,\ldots,X_m$ have the C1-property, then we use sorting 
strategy to decide whether $I^\pi$ is a YES-instance. 
According to Lemma~\ref{lem2:automorphism} (and because equality
up to renumbering is an equivalence relation), this is the case 
iff $I$ is a YES-instance.
\end{proof}

The following example\footnote{Example~\ref{ex:C1P} below is a slight
modification of the example shown in Fig.~8.3 in~\cite{G1980}.}
shows that there exists a collection $X_1,\ldots,X_m \seq [n]$ 
with the C1-property on which a direct application of the sorting 
strategy may fail:

\begin{example} \label{ex:C1P}
Consider the collection
\[ 
X_1 = \{1,4\}\ ,\ X_2 = \{2,4\}\ ,\  X_3 = \{2,3,5\}\ ,\  X_4 = \{1,2,4,5\} 
\enspace .
\]
Set $I = (5,4,X_1,X_2,X_3,X_4,\tau)$ for $\tau = (1,3,4,2,5)$.
An inspection of
\[
(1,3,4,2,5) \transition{1,4} (1,3,4,2,5) \transition{2,4} (1,3,4,2,5)
\transition{2,3,5} (1,4,3,2,5) \transition{1,2,4,5} (1,2,3,4,5) 
\]
shows that $I$ is a YES-instance of V2-WPPSG.
Choose now $\pi = (1,4,2,5,3)$ so that $\pi^{-1} = (1,3,5,2,4)$.
The sets
\[
X_1\pi^{-1} = \{1,2\}\ ,\ X_2\pi^{-1} = \{2,3\}\ ,\  X_3\pi^{-1} = \{3,4,5\}\ ,\  
X_4\pi^{-1} = \{1,2,3,4\} 
\]
consist of consecutive integers, respectively. Thus $\{X_1,X_2,X_3,X_4\}$
has the C1-property. Thus any $I = (5,4,X_1,X_2,X_3,X_4,\tau)$
with $\tau \in \cS_5$ is an input instance of WPPSG$_1$ and the
sorting strategy is fine when applied to $I^\pi$. However, 
for $\tau = (1,3,4,2,5)$, the sorting strategy fails on $I$ because 
it leads to the following sequence of permutations:
\[
(1,3,4,2,5) \transition{1,4} (1,3,4,2,5) \transition{2,4} (1,2,4,3,5)
\transition{2,3,5} (1,2,4,3,5) \transition{1,2,4,5} (1,2,4,3,5)
\enspace .
\]
\end{example}

\section{A More General Subproblem of WPPSG} 
\label{sec:large-subproblem}

In Section~\ref{subsec:transformation}, transformations of input
instances are brought into play. They will have the property that
the transformation $I'$ of an instance $I$ is a YES-instance
of WPPSG iff the original \hbox{instance $I$} is a YES-instance. 
In Section~\ref{subsec:nice-instances} the subproblem WPPSG$_2$
is defined as the problem WPPSG restricted to ``nice instances''
where an instance is called \emph{nice} if it can be transformed 
to some input instance of WPPSG$_0$. As we shall see, instances 
of WPPSG$_1$ are very special cases of nice instances. 
Therefore the problem WPPSG$_2$ is significantly more general 
than WPPSG$_1$. It will turn out furthermore that niceness 
can be tested efficiently. The main building blocks of such 
a test are presented in Section~\ref{subsec:niceness-test}.

\subsection{Transformation of Input Instances} 
\label{subsec:transformation}

As usual, let $I = (n,m,X_1,\ldots,X_m,\tau)$ an arbitrary but fixed instance 
of WPPSG.  Let $j' \in [m]$ and let $\varphi \in \cS_{X_{j'}}$.
The \emph{$(j',\varphi)$-transformation of $I$}, denoted by $I[j',\varphi]$
in the sequel, is given by $I[j',\varphi] = (n,m,X'_1,\ldots,X'_m,\tau')$ 
where  
\begin{equation} \label{eq:admissibility}
\tau' = \tau \varphi\ \mbox{ and }\ 
X'_j = \left\{ \begin{array}{ll}
        X_j & \mbox{if $j \le j'$} \\
        X_j\varphi & \mbox{if $j > j'$}
       \end{array} \right. \enspace .
\end{equation}
To avoid unnecessary case distinctions, we set $I[0,\varphi] = I^\varphi$
and $X_0 = [n]$. 
Note that $I[j,\id] = I$ and $(I[j,\varphi])[j,\psi] = I[j,\varphi\psi]$.
It follows that $(I[j,\varphi])[j,\varphi^{-1}] = I[j,\id] = I$.
We call $I'$ a \emph{valid elementary transformation of $I$} and 
write $I \ra I'$ if $I' = I[j,\varphi]$ for some $j \in [m]\cup\{0\}$ 
and some $\varphi \in \cS_{X_j}$. We refer to $\varphi \in \cS_{X_j}$
as the \emph{validity condition}. Note that $\ra$ is a reflexive and 
symmetric relation.

\begin{lemma} \label{lem:elementary-transformation}
With the above notations the following hold.
If $I = (n,m,X_1,\ldots,X_m,\tau)$ is a YES-instance of WPPSG, 
then every valid elementary transformation of $I$ is a YES-instance too.
\end{lemma}

\begin{proof}
Choose an arbitrary valid elementary transformation $I[j',\varphi]$ of $I$.
Since Lemma~\ref{lem2:automorphism} covers the case $j'=0$, we may assume 
that $j'\ge1$. Suppose that there 
exist $\sigma_1\in\cS_{X_1},\ldots,\sigma_m\in\cS_{X_m}$ such 
that $\tau = \sigma_1 \ldots \sigma_m$. For $j = 1,\ldots,m$, 
we define
\begin{equation} \label{eq:transformation}
\sigma'_j = \left\{ \begin{array}{ll}
           \sigma_j & \mbox{if $j<j'$} \\
           \sigma_{j'} \varphi & \mbox{if $j=j'$} \\
           \varphi^{-1} \sigma_j \varphi & \mbox{if $j>j'$}
            \end{array} \right. \enspace .
\end{equation}
The following calculation shows that $\tau' = \tau\varphi$ can be
written as a product of the permutations $\sigma'_1,\ldots\sigma'_m$:
\begin{eqnarray*}
\tau \varphi & = &
\sigma_1 \ldots \sigma_{j'-1} \sigma_{j'} 
\sigma_{j'+1} \ldots \sigma_m \varphi \\ 
& = & 
\sigma_1 \ldots \sigma_{j'-1} (\sigma'_{j'} \varphi^{-1}) 
\sigma_{j'+1} \ldots \sigma_m \varphi) \\
& = &
\sigma_1 \ldots \sigma_{j'-1} (\sigma'_{j'} \varphi^{-1})
(\varphi \sigma'_{j'+1} \varphi^{-1})
\ldots (\varphi \sigma'_m \varphi^{-1}) \varphi\ =\    
\sigma'_1 \ldots\sigma'_m \enspace .
\end{eqnarray*}
Remember that $\varphi \in \cS_{X_{j'}}$ and $\sigma_j \in \cS_{X_j}$ 
for any choice of $j \in [m]$. 
We claim that $\sigma'_j \in \cS_{X'_j}$ (with $\sigma'_j$
as given by~(\ref{eq:transformation}) and $X'_j$ as given
by~(\ref{eq:admissibility})). For $j<j'$, 
this holds because $\sigma'_j = \sigma_j$ and $X'_j = X_j$.
For $j = j'$, this holds because  $X'_{j'} = X_{j'}$ and
$\sigma'_{j'} = \sigma_{j'}\varphi$ is the composition of two 
permutations in $\cS_{X'_{j'}}$ and therefore itself a permutation 
in $\cS_{X'_{j'}}$. For $j>j'$, it holds 
because $\sigma'_j = \varphi^{-1}\sigma_j\varphi$ 
and $X'_j = X_j\varphi$ so that we may apply Lemma~\ref{lem1:automorphism}.
This discussion shows that $I[j',\varphi]$ is a YES-instance.
\end{proof}

We call $I'$ a \emph{transformation of $I$} and write $I \transition{*} I'$, 
if $I'$ can be obtained from $I$ by a chain of valid elementary transformations. 
Note that $\transition{*}$ is an equivalence relation. 
From Lemma~\ref{lem:elementary-transformation}, we immediately obtain
the following result:

\begin{lemma} \label{lem:transformation}
If $I$ is a YES-instance of WPPSG, then every transformation of $I$
is a YES-instance too.
\end{lemma}

\noindent
The concatenation of two valid elementary transformations is commutative
in the following sense:
\begin{lemma} \label{lem:commutativity}
Let $I = (n,m,X_1,\ldots,X_m,\tau)$ be an instance of WPPSG.
Let $0 \le j_1 < j_2 \le m$ and let $\varphi_2$ and $\varphi'_2$
be permutations on $[n]$ such 
that $\varphi'_2 = \varphi_1^{-1} \varphi_2 \varphi_1$.
Suppose that $\varphi_1 \in \cS_{X_{j_1}}$. Suppose furthermore
that $\varphi_2 \in \cS_{X_{j_2}}$.\footnote{According to 
Lemma~\ref{lem1:automorphism}, this implies
that $\varphi'_2 \in \cS_{X_{j_2}\varphi_1}$.}  
Then $(I[j_1,\varphi_1])[j_2,\varphi'_2] = (I[j_2,\varphi_2])[j_1,\varphi_1]$.
\end{lemma}

\begin{proof}
The assumptions $\varphi_1 \in \cS_{X_{j_1}}$ 
and $\varphi_2 \in \cS_{X_{j_2}}$ make sure that all validity conditions
which are relevant for the lemma are satisfied. Choose 
now $X_1^\ua,\ldots,X_m^\ua,\tau^\ua$ and $X_1^\da,\ldots,X_m^\da,\tau^\da$ 
such that 
\begin{eqnarray*}
(I[j_1,\varphi_1])[j_2,\varphi'_2] & = & (n,m,X^\ua_1,\ldots,X^\ua_m,\tau^\ua) \\
(I[j_2,\varphi_2])[j_1,\varphi_1] & = & (n,m,X^\da_1,\ldots,X^\da_m,\tau^\da)
\end{eqnarray*}
We have to show that $\tau^\ua = \tau^\da$ and $X_j^\ua = X_j^\da$
for $j = 1,\ldots,m$. Let's discuss first the case $j_1 \ge 1$.
It follows that $\tau^\da = \tau\varphi_2\varphi_1$
and $\tau^\ua = \tau\varphi_1\varphi'_2$. Because 
of $\varphi'_2 = \varphi_1^{-1} \varphi_2 \varphi_1$,
we get $\tau^\da = \tau^\ua$. Clearly $X_j^\da = X_j = X_j^\ua$
for $j \le j_1$ and $X_j^\da = X_j\varphi_1 = X_j^\ua$
for $j_1 < j \le j_2$. Suppose now that $j>j_2$. 
Then we obtain $X_j^\da = X_j\varphi_2\varphi_1$
and $X_j^\ua = X_j\varphi_1\varphi'_2$. Expanding $\varphi'_2$
in terms of $\varphi_2$, we see that $X_j^\da = X_j^\ua$. \\
Assume now that $j_1=0$. It follows 
that $\tau^\da = \varphi_1^{-1}\tau\varphi_2\varphi_1$
and $\tau^\ua = \varphi_1^{-1}\tau\varphi_1\varphi'_2$. Because
of $\varphi'_2 = \varphi_1^{-1} \varphi_2 \varphi_1$,
we get $\tau^\da = \tau^\ua$. The calculations which show 
that $X_j^\da = X_j^\ua$ are the same as as above and they
follow the same case distinction except that the case $j \le j_1$ 
becomes vacuous for $j_1=0$.
\end{proof}

\noindent
A chain of valid elementary transformations from $I$ to $I' = I_m$
of the form
\begin{equation} \label{eq:ascending-chain}
I \ra I_1 = I^{\pi_1} \ra I_2 = I_1[1,\sigma_1] \ra\ldots\ra  
I_m = I_{m-1}[m-1,\sigma_{m-1}]
\end{equation}
is called \emph{an ascending chain}. 
The tupel $(\pi_1,\sigma_1,\ldots,\sigma_{m-1})$ is called
its \emph{description}. In view of the
identity $(I[j,\varphi])[j,\psi] = I[j,\varphi\psi]$ 
and of Lemma~\ref{lem:commutativity}, we obtain the 
following result:

\begin{corollary} \label{cor:normal-form}
If $I \transition{*} I'$, then there exists an ascending chain 
that transforms $I$ into $I'$.
\end{corollary}

\subsection{WPPSG Restricted to Nice Input Instances}
\label{subsec:nice-instances}

Call an input instance $I$ of WPPSG \emph{nice} if $I \transition{*} I'$
holds for some instance $I'$ of WPPSG$_0$. Let WPPSG$_2$ be the restriction 
of WPPSG to nice instances. Note that, for obvious reasons, every input 
instance of WPPSG$_1$ is nice. Hence WPPSG$_1$ is a subproblem of WPPSG$_2$. 
We will see in Sections~\ref{subsec:niceness-test} 
and~\ref{subsec:overcome} that we can efficiently decide 
whether $I$ is nice and, if applicable, that we can efficiently compute 
the description of an ascending chain leading from $I$ to some instance $I'$ 
of WPPSG$_0$. Therefore WPPSG$_2$ can be solved efficiently as follows:
\begin{enumerate}
\item 
Test whether $I$ is nice. Reject $I$ if this is not the case
and, otherwise, compute the description of an ascending chain 
from $I$ to $I'$ for some instance $I'$ of WPPSG$_0$.
\item
Use the description of this ascending chain to actually compute $I'$.
\item
Decide by means of the sorting strategy whether $I'$ is a YES-instance.
\end{enumerate}

\subsection{The Main Building Blocks of a Test for Niceness} 
\label{subsec:niceness-test}

A central tool in the sequel is the procedure REDUCE-EXPAND which,
given sets $X_1,\ldots,X_m \seq [n]$, works as follows:
\begin{description}
\item[Initialization.]
$X_0 := [n]$ and $\Pi_0 := \cS_n$.
\item[Loop.]
For $j = 1,\ldots,m$:
\begin{description}
\item[Expansion Step:]
$\Pi'_j := \{\pi\sigma: \pi \in \Pi_{j-1} \mbox{ and } \sigma \in \cS_{X_{j-1}\pi}\}$.
\item[Reduction Step:]
$\Pi_j := \{\pi \in \Pi'_j: \mbox{ $X_j\pi$ is a set of consecutive integers}\}$.
\end{description}
\end{description}
By the first execution of the for-loop, we get
\begin{equation} \label{eq:first-step}
\Pi'_1 := \cS_n\ \mbox{ and }\ \Pi_1 := 
\{\pi \in \cS_n: \mbox{$X_1\pi$ is a set of consecutive integers}\}
\enspace .
\end{equation}
We say that the sets $X_1,\ldots,X_m \seq [n]$ have the
\emph{Weak Consecutive Ones Property (WC1P)} if REDUCE-EXPAND
with input $X_1,\ldots,X_m$ computes 
sets $\Pi_0,\Pi'_1,\Pi_1,\ldots,\Pi'_m,\Pi_m$ such that $\Pi_m \neq \eset$
(which implies that all these sets are non-empty). 
We briefly note that the removal of all expansion steps would have
the consequence that $\Pi_m \neq \eset$ iff $X_1,\ldots,X_m$
have the C1-property. Since the removal of expansion steps increases
the chances to arrive at some empty set $\Pi_j$, it follows that
the C1-property implies the weak C1-property. It does not come as
a surprise that there exist collections $X_1,\ldots,X_m$ having 
the \emph{weak} C1-property but not the C1-property itself.
We will illustrate this by an example after we have proven the
following result:

\begin{theorem} \label{th:reduce-expand}
Let $I = (n,m,X_1,\ldots,X_m,\tau)$ be an instance of WPPSG
and let $\Pi_0,\Pi'_1,\Pi_1 ,\ldots, \Pi'_m,$ $\Pi_m$ be the sequence
of permutation sets computed by the procedure REDUCE-EXPAND with 
input $X_1,\ldots,X_m$. With these notations, we have that $I$ is 
nice iff $\Pi_m \neq \eset$.
\end{theorem}

\begin{proof}
Suppose first that $\Pi_m \neq \eset$. We will prove the niceness 
of $I$ by finding a description of an ascending chain from $I$ to 
some instance of WPPSG$_0$. We start by choosing an arbitrary 
but fixed permutation $\pi_m \in \Pi_m$. Since $\Pi_m \seq \Pi'_m$, 
we may also choose permutations $\pi_{m-1} \in \Pi_{m-1}$ 
and $\sigma_{m-1} \in \cS_{X_{m-1}\pi_{m-1}}$
such that $\pi_m = \pi_{m-1}\sigma_{m-1}$. Using this reasoning 
iteratively, we see that there exist 
$\pi_m \in \Pi_m, \pi_{m-1} \in \Pi_{m-1}, \ldots, \pi_1 \in \Pi_1$
and $\sigma_{m-1} \in \cS_{X_{m-1}\pi_{m-1}}, \ldots, \sigma_1 \in \cS_{X_1\pi_1}$
which satisfy $\pi_j = \pi_{j-1}\sigma_{j-1}$ for $j = 2,\ldots,m$.
This implies that 
\begin{equation} \label{eq:recursion}
\forall j = 2,\ldots,m: \pi_j = \pi_1\sigma_1\ldots\sigma_{j-1} \enspace .
\end{equation}
Let's interpret $(\pi_1,\sigma_1,\ldots,\sigma_{m-1})$ as the description
of an ascending chain of the form (\ref{eq:ascending-chain}).
The chain starts at $I$, passes through $I_1,\ldots,I_{m-1}$ 
and ends in $I_m$. We denote the $j$-th specification set of $I_{j_0}$
\hbox{by $X_j^{j_0}$}. We have two show two things: 
\begin{enumerate}
\item
The validity conditions in~(\ref{eq:ascending-chain}) are satisfied, i.e.,
\begin{equation} \label{eq:validity} 
\forall j=1,\ldots,m-1: \sigma_j \in \cS_{X_j^j} \enspace .
\end{equation}
Since we already know that $\sigma_j \in \cS_{X_j\pi_j}$,
it suffices to show that $X_j^j = X_j\pi_j$.
\item
$I_m$ is an input instance of WPPSG$_0$, i.e., 
$X_1^m,\ldots,X_m^m$ are sets of consecutive integers.
\end{enumerate}
Expanding the definition of $I^\pi_1$ and $I_j[j,\sigma_j]$, 
we see that the following holds for every choice of $j,j_0 \in [m]$:
\begin{equation} \label{eq:instances}
X_j^{j_0} = \left\{ \begin{array}{ll}
              X_j\pi_1\sigma_1\ldots\sigma_{j_0-1} = X_j\pi_{j_0} & 
                \mbox{if $j \ge j_0$} \\
              X_j^j = X_j\pi_j & \mbox{if $j < j_0$} 
           \end{array} \right.  \enspace .
\end{equation}
Setting $j_0 = m$ or $j_0 = j$, we obtain 
\begin{equation} \label{eq:final-instance}
X_j^m = X_j^j = X_j\pi_j = X_j\pi_1\sigma_1\ldots\sigma_{j-1} 
\end{equation} 
for $j = 1,\ldots,m$. From $\pi_j \in \Pi_j$ and the definition 
of $\Pi_j$, it follows that $X_1^m,\ldots,X_m^m$ are sets of 
consecutive integers. This concludes the proof that $I$ is nice. \\
Now we turn our attention to the other proof direction. Suppose that 
there is an ascending chain of the form~(\ref{eq:ascending-chain})
such that the corresponding validity conditions are satisfied and 
such that $I_m$ is an instance of WPPSG$_0$. We have to show
that $\Pi_m \neq \eset$. For $j = 2,\ldots,m$, 
we set $\pi_j := \pi_{j-1}\sigma_{j-1}$ which implies 
that $\pi_j = \pi_1\sigma_1\ldots\sigma_{j-1}$. As before, we denote $X_j^{j_0}$
the $j$-th specification set in the instance $I_{j_0}$. As in the first part
of the proof, the set $X_j^{j_0}$ satisfies~(\ref{eq:instances}). 
Since $I_m$ is an instance of WPPSG$_0$, it follows 
that $X_j^j = X_j\pi_j = X_j^m$ is a set of consecutive integers. 
Hence $X_1\pi_1 = X_1^1$ is a set of consecutive integers.
Because of~(\ref{eq:first-step}), we obtain $\pi_1 \in \Pi_1$.
\begin{description}
\item[Claim:]
For $j=1,\ldots,m$, we have that $\pi_j \in \Pi_j$.
\item[Proof of the Claim:]
We have verified the claim for $j=1$ already. Thus we may focus on
the inductive step from $j$ to $j+1$. Since the validity conditions are
satisfied, we have that $\sigma_j \in \cS_{X_j^j} = \cS_{X_j\pi_j}$. 
Hence $\pi_{j+1} = \pi_j\sigma_j \in \Pi'_{j+1}$. Moreover, 
since $X_{j+1}\pi_{j+1}$ is a set of consecutive integers,
we have that $\pi_{j+1} \in \Pi_{j+1}$, which completes the inductive proof
of the claim. 
\end{description}
The punchline of this discussion is that $\pi_m \in \Pi_m$ 
so that $\Pi_m \neq \eset$.
\end{proof}

\begin{example} \label{ex:WC1P}
Consider an instance $I$ of WPPSG with 
specification sets $X_1,X_2,X_3,X_4 \seq [5]$ given by
\[
X_1 = \{1,2\}\ ,\ X_2 = \{1,3\}\ ,\ X_3 = \{1,4\}\ ,\ X_4 = \{1,5\}
\enspace .
\]
It was mentioned in~\cite{G1980} (and it is easy to check) that 
the sets $X_1,X_2,X_3,X_4$ do not have the C1-property. We will
prove now that these sets have the weak C1-property. To this end,
we set $\pi_1 = \id$, $\sigma_1 = \sprod{1}{2}$, $\sigma_2 = \sprod{2}{3}$
and $\sigma_3 = \sprod{3}{4}$. Consider the ascending chain with 
description $(\pi_1,\sigma_1,\sigma_2,\sigma_3)$. We use the notation 
from the proof of Theorem~\ref{th:reduce-expand}. Because $\pi_1 = \id$,
an application of~(\ref{eq:final-instance}) yields that
\[
\begin{array}{cccclcc}  
X_1^4 & = & X_1^1 & = & X_1 & = & \{1,2\} \\
X_2^4 & = & X_2^2 & = & X_2 \sigma_1 & = & \{2,3\} \\
X_3^4 & = & X_3^3 & = & X_3 \sigma_1 \sigma_2 & = & \{3,4\} \\
X_4^4 & = & & = & X_4 \sigma_1 \sigma_2 \sigma_3 & = & \{4,5\}
\end{array}
\enspace .
\]
Note that, for each $j \in [3]$, we have that $X_j^j = X_j^4 = \{j,j+1\}$
and $\sigma_j = \sprod{j}{j+1} \in \cS_{X_j^j}$. It follows that the 
validity conditions for the ascending chain with 
description $(\pi_1,\sigma_1,\sigma_2,\sigma_3)$ are satisfied. 
Moreover this chain ends in an instance from WPPSG$_0$, namely in the 
instance with specification sets $X_1^4,X_2^4,X_3^4,X_4^4$. Thus $I$ is nice 
and, by virtue of Theorem~\ref{th:reduce-expand}, we may conclude
that $\Pi_4 \neq \eset$. In other words, $I$ has the weak C1-property.
\end{example}

In order to obtain an algorithm that solves WPPSG$_2$ efficiently,
one has to overcome the following obstacles.
\begin{description}
\item[Obstacle 1:] 
A succinct representation for (possibly exponentially large)
subsets of $\cS_n$ is needed.
Here we make use of PQ-trees which are well known from the
work of Booth and Lueker~\cite{BL1976}. 
See Section~\ref{subsec:pq-trees}.
\item[Obstacle 2:]
Booth and Lueker~\cite{BL1976} presented an efficient procedure
which implements the reduction step, i.e., it transforms a
PQ-tree which represents $\Pi'_j$ into a PQ-tree which 
\hbox{represents $\Pi_j$}. We have to present an efficient procedure which 
implements the expansion step, i.e., it transforms a PQ-tree which 
represents $\Pi_{j-1}$ into a PQ-tree which represents $\Pi'_j$. 
Iterated application of these two procedures obviously leads to 
an efficient implementation of the procedure REDUCE-EXPAND.
\item[Obstacle 3:]
Let $T_j$ (resp.~$T'_j$) denote a PQ-tree which represents $\Pi_j$
(resp.~$\Pi'_j$) and assume that the trees $T_1,\ldots,T_m$
and $T'_2,\ldots,T'_m$ are given. We have to present an efficient 
procedure which does the following. 
Given a permutation $\pi_j \in \Pi_j$, it computes 
permutations $\pi_{j-1} \in \Pi_{j-1}$ 
and $\sigma_{j-1} \in \cS_{X_{j-1}\pi_{j-1}}$ such 
that $\pi_j = \pi_{j-1}\sigma_{j-1}$. 
As should become clear from the proof of Theorem~\ref{th:reduce-expand}, 
iterated application of this procedure yields the 
description $(\pi_1,\sigma_1,\ldots,\sigma_{m-1})$ of an ascending 
chain which starts at $I$ and ends at some input instance of WPPSG$_0$. 
\end{description}
As for overcoming Obstacles 2 and 3, see Section~\ref{subsec:overcome}
and, in particular, Corollary~\ref{cor:flatten}.

\section{An Efficient Algorithm for WPPSG$_2$} 
\label{sec:wppsg2-algorithm}

In Section~\ref{subsec:pq-trees}, we recall some central definitions 
from~\cite{BL1976} and state some useful facts.
In Section~\ref{subsec:overcome}, we show how to overcome the obstacles
that we mentioned at the end of Section~\ref{sec:large-subproblem}.
We are then ready for the main result of this paper. The latter
is presented in Section~\ref{subsec:main-result}.

\subsection{PQ-Trees} \label{subsec:pq-trees}

\emph{PQ-trees over the universe $[n]$} are rooted
ordered trees whose leaves are elements of $[n]$ and whose internal
nodes are distinguished as either P- or Q-nodes.\footnote{Booth
and Lueker used the universe $\{u_1,\ldots,u_n\}$
instead of $[n]$.} 
In figures, a P-node is drawn as a circle with its children below it
while a Q-node is drawn as a rectangle with its children below it.
A PQ-tree $T$ is \emph{proper} if the following hold:
\begin{enumerate}
\item
Every element in $[n]$ appears precisely once as a leaf in $T$.
\item
Each P-node has at least $2$ and each Q-node has at least $3$ children.
\end{enumerate}
In the sequel, we only talk about proper PQ-trees without explicitly
mentioning this again. The symbol $T$ is reserved to denote such trees.
Reading the leaves in $T$ from left to right yields the \emph{frontier} 
of $T$, which is denoted by $\FRONT(T)$. 
We identify $\FRONT(T)$ with a permutation on $[n]$ in the usual way:
if $\FRONT(T) = (u_1 \ldots u_n) \in [n]^n$, then it equals the permutation 
which maps $i$ to $u_i$. Hence the inverse of $\FRONT(T)$, 
denoted by $\FRONT(T)^{-1}$, maps $u_i$ to $i$ for $i = 1,\ldots,n$.

A PQ-tree $T'$ is said to be \emph{equivalent to PQ-tree $T$} if it 
can be obtained by the following \emph{permitted modifications} of $T$,
which may take place at every internal node:
\begin{description}
\item[Modification at a P-node:]
The children of a P-node can be arbitrarily permuted (including, 
as a special case, the option to leave the order of the children
unchanged).
\item[Modification at a Q-node:]
The children of a Q-node are either left in their original order
or they are put in reverse order.
\end{description}
If $T$ is a PQ-tree and $x$ is a node in $T$, then $T_x$
denotes the subtree of $T$ that is rooted \hbox{at $x$}. 
A permutation $\pi$ is \emph{consistent with $T$} if $\pi$ equals
the frontier of one of the trees which are equivalent to $T$. The 
set of permutations that are consistent with $T$ is denoted by $\CONS(T)$. 

\begin{definition} \label{def:consistency}
The \emph{consistency problem for PQ-trees} is the following problem:
given a PQ-tree $T$ with $n$ leaves and a permutation $\pi \in \cS_n$, 
decide whether $\pi \in \CONS(T)$ and, if applicable, construct a tree 
that is equivalent to $T$ and has frontier $\pi$.
\end{definition}

\begin{lemma}[Common Knowledge] \label{lem:consistency}
The consistency problem for PQ-trees can be solved in time $O(n)$.
\end{lemma}

We suspect that this lemma is known but, for sake of completeness,
its proof will be presented in Section~\ref{sec1:app} of the appendix.

\noindent
A PQ-tree $T$ is said to \emph{represent} the set
\[
\CONS^{-1}(T) := \{\pi \in \cS_n: \pi^{-1} \in \CONS(T)\} \seq \cS_n
\enspace .
\]

Let $X \seq [n]$ be of size at least $2$ and let $T$ be a PQ-tree over $[n]$.
A subtree $T_x$ of $T$ is called an \emph{$X$-tree} if the set of leaves 
in $T_x$ equals $X$.

Let $x_1,\ldots,x_r$ be siblings in $T$ which follow one another successively
in the order from left to right and let $x$ denote their common parent. 
The forest $F$ formed by $T_{x_1},\ldots,T_{x_r}$ is called an \emph{$X$-forest} 
if the set of leaves in $F$ equals $X$. If $F$ is an $X$-forest, then $T_x$ 
is said to \emph{contain an $X$-forest}. 

\begin{lemma}[Common Knowledge] \label{lem:consecutive-integers}
Let $X \seq [n]$ be of size at least $2$ and let $T$ be a PQ-tree such that
\begin{equation} \label{eq:consecutive-integers}
\CONS(T) \seq \{\pi\in\cS_n:
\mbox{the elements of $X$ occur consecutively within $\pi$}\} \enspace .
\end{equation}
Let $x$ be the youngest common ancestor of all $X$-leaves in $T$.
Then the following hold:
\begin{enumerate}
\item
If $x$ is a P-node, then $T_x$ is an $X$-tree.
\item
If $x$ is a $Q$-node, then $T_x$ contains an $X$-forest.
\end{enumerate}
\end{lemma}

We suspect that this lemma is known but, for sake of completeness,
its proof will be presented in Section~\ref{sec2:app} of the appendix.


\subsection{Overcoming Obstacles 2 and 3} \label{subsec:overcome}

Let $\Pi_0,\Pi'_1,\Pi_1,\ldots,\Pi'_m,\Pi_m$ be the sets returned
by the procedure REDUCE-EXPAND when the latter is applied 
to $X_1,\ldots,X_m \seq [n]$. Our goal is to efficiently 
construct PQ-trees $T_0,T'_1,T_1,\ldots,T'_m,T_m$
such that $T_j$ represents $\Pi_j$ and $T'_j$ represents $\Pi'_j$.
$\Pi_0 = \cS_n$ is represented by the \emph{universal PQ-tree}
with a single internal node which is a P-node that has children $1,\ldots,n$. 
Let's denote this tree \hbox{by $T_0$}. It suffices now to show the following:
\begin{description}
\item[Claim 1:]
Given a PQ-tree $T_{j-1}$ which represents $\Pi_{j-1}$,
one can construct in time $O(n)$ a PQ-tree $T'_j$ which represents $\Pi'_j$.
\item[Claim 2:]
Given a PQ-tree $T'_j$ which represents $\Pi'_j$,
one can construct in time $O(n)$ a PQ-tree $T_j$ which represents $\Pi_j$.
\end{description}
Claim 2 is covered by the work in~\cite{BL1976}. We may therefore focus 
on the proof of Claim 1.
In the sequel, we verify Claim 1 by proving a sligtly stronger 
statement. Details follow.


Suppose that a PQ-tree $T$ and a set $X \seq [n]$ of size at least $2$
satisfy condition~(\ref{eq:consecutive-integers}). Then we define
by $\FLAT(T,X)$ the PQ-tree which is obtained from $T$ as follows
(with the notation from Lemma~\ref{lem:consecutive-integers}):
\begin{itemize}
\item
If $x$ is a P-node, then replace the $X$-tree $T_x$ by another $X$-tree $T'_x$
of height $1$: the root of $T'_x$ is still the P-node $x$ but now all leaves 
of $T_x$ are made children of $x$. Note that $T$ without its subtree $T_x$ 
is exactly the same tree as $\FLAT(T,X)$ without its subtree $T'_x$.
\item
If $x$ is a Q-node, then replace the $X$-forest $F$ that is contained
in $T_x$ by the $X$-tree $T'_{x'}$ of height $1$ that is obtained from $F$
as follows: merge the roots of the trees in $F$ to a single new P-node $x'$
and make all leaves in $F$ children of $x'$. Note that $T$ without
its sub-forest $F$ is exactly the same tree as $\FLAT(T,X)$ without
its subtree $T'_{x'}$ 
\end{itemize}
Note that $\FLAT(T,X)$ contains an $X$-tree of height $1$ as a subtree
(just like any PQ-tree that is equivalent to $\FLAT(T,X)$). 
This has the implication that, for every 
permutation $\pi' \in \CONS^{-1}(\FLAT(T,X))$, the set $X\pi'$ 
consists of consecutive integers.

\begin{lemma} \label{lem:flatten}
Suppose that a PQ-tree $T$ and a set $X \seq [n]$ of size at least $2$
satisfy condition~(\ref{eq:consecutive-integers}). Then the following 
holds:
\begin{equation} \label{eq:X-subset}
\CONS^{-1}(\FLAT(T,X)) = \{\pi \sigma: \pi \in \CONS^{-1}(T) \mbox{ and } 
\sigma \in \cS_{X\pi}\} \enspace .
\end{equation}
Moreover, given a permutation $\pi' \in \CONS^{-1}(\FLAT(T,X))$
and the trees $T$ and $\FLAT(T,X)$, one can compute in time $O(n)$ 
permutations $\pi \in \CONS^{-1}(T)$ and $\sigma \in \cS_{X\pi}$ 
such that $\pi' = \pi\sigma$.
\end{lemma}

\begin{proof}
We use the notation from Lemma~\ref{lem:consecutive-integers} and 
from the above definition of $\FLAT(T,X)$. Remember that $x$ denotes 
the youngest common ancestor of all $X$-leaves in $T$. 
Suppose that $x$ is a P-node.\footnote{The proof under the assumption 
that $x$ is a Q-node is similar.} According to 
Lemma~\ref{lem:consecutive-integers}, $T_x$ is an $X$-tree. 
For sake of brevity, set $F(T,X) := \FLAT(T,X)$. \\
We first show that $\CONS^{-1}(F(T,X))$ is contained in the set on the 
right-hand side of~(\ref{eq:X-subset}). See Fig.~\ref{fig:flatten} 
(with $X$-leaves in $T$ and $F(T,X)$ highlighted in bold) for an 
illustration of the following reasoning. Choose an arbitrary but fixed 
permutation $\pi' \in \CONS^{-1}(F(T,X))$. Then there exist modifications 
at the internal nodes of $F(T,X)$ such that $F(T,X)$ is transformed into 
an equivalent PQ-tree, say $F(T,X)''$, whose frontier equals the inverse 
of $\pi'$. We can think of $F(T,X)''$ as being obtained in two stages:
\begin{itemize}
\item 
In Stage~1, we modify $F(T,X)$ only at $P$- or $Q$-nodes outside $F(T,X)_x$
and obtain a modified tree, say $F(T,X)'$, which still has $F(T,X)_x$ as a 
subtree.
\item
In Stage 2, we obtain $F(T,X)''$ by making one final modification 
at the root $x$ of $F'(T,X)_x$. 
\end{itemize}
Note that in Stage 2 nothing is changed except for the internal ordering
of the $X$-leaves in $F'(T,X)$. Now, since $T$ and $F(T,X)$ coincide outside their 
respective subtrees $T_x$ and $F(T,X)_x$, we can make all modifications 
falling into Stage 1 also in $T$. Denote the resulting tree by $T'$ 
and let $\pi$ be the inverse of $\FRONT(T')$. It follows 
that $\pi \in \CONS^{-1}(T)$ and $u\pi = u\pi'$ for each $u \in [n] \sm X$.
Set $\sigma = \pi^{-1}\pi'$ so that $\pi\sigma = \pi'$. 
If $u\pi \notin X\pi$, then $u \in [n] \sm X$ and $u\pi\sigma = u\pi' = u\pi$.
Thus $\sigma \in \cS_{X\pi}$. This completes the proof of ``$\seq$''
in~(\ref{eq:X-subset}). \\
We move on and prove that $\pi \in \CONS^{-1}(T)$ and $\sigma \in \cS_{X\pi}$
with $\pi' = \pi\sigma$ can be computed in time $O(n)$:
\begin{itemize}
\item
According to Lemma~\ref{lem:consistency}, the consistency problem 
for PQ-trees (see Definition~\ref{def:consistency}) can be solved 
in time $O(n)$. It follows that the PQ-tree $F(T,X)''$ whose frontier 
equals the inverse of $\pi'$ (as well as the corresponding modifications 
of $F(T,X)$) can be computed in time $O(n)$. 
\item
Applying the modifications which fall into Stage 1 to the PQ-tree $T$,
we get in time $O(n)$ the tree $T'$ and the 
permutation $\pi = \FRONT(T')^{-1}$. 
\item
As argued above, a suitable permutation $\sigma \in \cS_{X\pi}$
is now obtained by setting $\sigma := \pi'\pi^{-1}$.
\end{itemize}
We still have to show that $\CONS^{-1}(F(T,X))$ is a superset of the set 
on the right-hand side of~(\ref{eq:X-subset}). Choose an arbitrary
but fixed permutation of the form $\pi \sigma$ 
with $\pi \in \CONS^{-1}(T,X)$ and $\sigma \in \cS_{X\pi}$. Then there 
exist modifications at the internal nodes of $T$ which transform $T$ 
into an equivalent PQ-tree $T''$ whose frontier equals $\pi^{-1}$. 
We can think of $T''$ as being obtained in two stages:
\begin{itemize}
\item
In Stage~1, we modify $T$ only at $P$- or $Q$-nodes outside $T_x$
and obtain a modified tree, say $T'$, which still has $T_x$ as a
subtree.
\item
In Stage 2, we obtain $T''$ by making some final modifications
at the internal nodes of $T'_x$.
\end{itemize}
Since $T$ and $F(T,X)$ coincide outside their respective subtrees $T_x$ 
and $F(T,X)_x$, we can make all modifications falling into Stage 1 also 
in $F(T,X)$. Denote the resulting tree by $F(T,X)'$
and let $\pi'$ be the inverse of $\FRONT(F(T,X)')$. 
In combination with $\sigma \in \cS_{X\pi}$, it follows
that $u\pi\sigma = u\pi = u\pi'$ holds for each $u \in [n] \sm X$. 
Hence that $X\pi\sigma = X\pi'$ but $\pi\sigma$ and $\pi'$ may induce
different permutations on $X$. However we may reorder the children
of $x$ in $F(T,X)'$ such that the resulting PQ-tree $F(T,X)''$
represents $\pi\sigma$. This completes the proof of ``$\supseteq$'' 
in~(\ref{eq:X-subset}).
\end{proof}

\begin{figure}[hbt]
        \begin{center}
\begin{tikzpicture}



        \def\x{-4} \def\y{0}

        \node[fill=cyan,rectangle, minimum width = 1cm] (a) at (\x,\y) {root of $T$};
        \node[fill=cyan,circle] (b) at (\x-3,\y-1) {};
        \node[fill=cyan,circle] (c) at (\x,\y-1) {$x$};
        \node[fill=cyan,rectangle,minimum width = 1cm] (d) at (\x+3,\y-1) {};
        \node[fill=white] (e) at (\x-3.5,\y-2) {$1$};
        \node[fill=white] (f) at (\x-2.5,\y-2) {$4$};
        \node[fill=cyan,rectangle,minimum width = 1cm] (g) at (\x-1,\y-2) {};
        \node[fill=white] (h) at (\x+1,\y-2) {\bf 9};
        \node[fill=white] (i) at (\x+2,\y-2) {$8$};
        \node[fill=white] (j) at (\x+3,\y-2) {$3$};
        \node[fill=white] (k) at (\x+4,\y-2) {$6$};
        \node[fill=white] (l) at (\x-2,\y-3) {\bf 2};
        \node[fill=white] (m) at (\x-1,\y-3) {\bf 7};
        \node[fill=white] (n) at (\x,\y-3) {\bf5};

        \draw (a)--(b); 
        \draw (a)--(c);
        \draw (a)--(d);
        \draw (b)--(e);
        \draw (b)--(f);
        \draw (c)--(g);
        \draw (c)--(h);
        \draw (d)--(i);
        \draw (d)--(j);
        \draw (d)--(k);
        \draw (g)--(l);
        \draw (g)--(m);
        \draw (g)--(n);

\def\x{4} \def\y{0}

        \node[fill=cyan,rectangle, minimum width = 1cm] (a') at (\x,\y) {root of $F(T,X)$};
        \node[fill=cyan,circle] (b') at (\x-3,\y-1) {};
        \node[fill=cyan,circle] (c') at (\x,\y-1) {$x$};
        \node[fill=cyan,rectangle,minimum width = 1cm] (d') at (\x+3,\y-1) {};
        \node[fill=white] (e') at (\x-3.5,\y-2) {$1$};
        \node[fill=white] (f') at (\x-2.5,\y-2) {$4$};
        \node[fill=white] (h') at (\x+1.25,\y-2) {\bf 9};
        \node[fill=white] (i') at (\x+2,\y-2) {$8$};
        \node[fill=white] (j') at (\x+3,\y-2) {$3$};
        \node[fill=white] (k') at (\x+4,\y-2) {$6$};
        \node[fill=white] (l') at (\x-1.5,\y-2) {\bf 2};
        \node[fill=white] (m') at (\x-0.65,\y-2) {\bf 7};
        \node[fill=white] (n') at (\x+0.35,\y-2) {\bf 5};

        \node (o) at (0.25,0.4) {};
        \node (p) at (0.25,-8.4) {};

        \draw (a')--(b'); 
        \draw (a')--(c');
        \draw (a')--(d');
        \draw (b')--(e');
        \draw (b')--(f');
        \draw (c')--(l');
        \draw (c')--(m');
        \draw (c')--(n');
        \draw (c')--(h');
        \draw (d')--(i');
        \draw (d')--(j');
        \draw (d')--(k');

        \draw (o)--(p);

\def\x{5} \def\y{-3}

        \node[fill=cyan,rectangle, minimum width = 1cm] (a') at (\x,\y) {root of $F(T,X)'$};
        \node[fill=cyan,circle] (b') at (\x+2.5,\y-1) {};
        \node[fill=cyan,circle] (c') at (\x,\y-1) {$x$};
        \node[fill=cyan,rectangle,minimum width = 1cm] (d') at (\x-3,\y-1) {};
        \node[fill=white] (e') at (\x+3,\y-2) {$1$};
        \node[fill=white] (f') at (\x+2,\y-2) {$4$};
        \node[fill=white] (h') at (\x+1,\y-2) {\bf 9};
        \node[fill=white] (i') at (\x-2,\y-2) {$8$};
        \node[fill=white] (j') at (\x-3,\y-2) {$3$};
        \node[fill=white] (k') at (\x-4,\y-2) {$6$};
        \node[fill=white] (l') at (\x-1,\y-2) {\bf 2};
        \node[fill=white] (m') at (\x-0.3,\y-2) {\bf 7};
        \node[fill=white] (n') at (\x+0.3,\y-2) {\bf 5};

        \draw (a')--(b'); 
        \draw (a')--(c');
        \draw (a')--(d');
        \draw (d')--(i');
        \draw (d')--(j');
        \draw (d')--(k');
        \draw (c')--(l');
        \draw (c')--(m');
        \draw (c')--(n');
        \draw (c')--(h');
        \draw (b')--(f');
        \draw (b')--(e');

\def\x{-3.5} \def\y{-4}

        \node[fill=cyan,rectangle, minimum width = 1cm] (a') at (\x,\y) {root of $T'$};
        \node[fill=cyan,circle] (b') at (\x+2.5,\y-1) {};
        \node[fill=cyan,circle] (c') at (\x,\y-1) {$x$};
        \node[fill=cyan,rectangle,minimum width = 1cm] (d') at (\x-3,\y-1) {};
        \node[fill=white] (e') at (\x+3,\y-2) {$1$};
        \node[fill=white] (f') at (\x+2,\y-2) {$4$};
        \node[fill=white] (h') at (\x+1,\y-2) {\bf 9};
        \node[fill=white] (i') at (\x-2,\y-2) {$8$};
        \node[fill=white] (j') at (\x-3,\y-2) {$3$};
        \node[fill=white] (k') at (\x-4,\y-2) {$6$};
        \node[fill=white] (l') at (\x-1.8,\y-3) {\bf 2};
        \node[fill=white] (m') at (\x-1,\y-3) {\bf 7};
        \node[fill=white] (n') at (\x-0.2,\y-3) {\bf 5};
        \node[fill=cyan,rectangle,minimum width = 1cm] (g') at (\x-1,\y-2) {};

        \draw (a')--(b'); 
        \draw (a')--(c');
        \draw (a')--(d');
        \draw (d')--(k');
        \draw (d')--(j');
        \draw (d')--(i');
        \draw (g')--(m');
        \draw (g')--(l');
        \draw (g')--(n');
        \draw (c')--(g');
        \draw (c')--(h');
        \draw (b')--(f');
        \draw (b')--(e');

\def\x{5} \def\y{-6}

        \node[fill=cyan,rectangle, minimum width = 1cm] (a') at (\x,\y) {root of $F(T,X)''$};
        \node[fill=cyan,circle] (b') at (\x+2.5,\y-1) {};
        \node[fill=cyan,circle] (c') at (\x,\y-1) {$x$};
        \node[fill=cyan,rectangle,minimum width = 1cm] (d') at (\x-3,\y-1) {};
        \node[fill=white] (e') at (\x+3,\y-2) {$1$};
        \node[fill=white] (f') at (\x+2,\y-2) {$4$};
        \node[fill=white] (h') at (\x+1,\y-2) {\bf 7};
        \node[fill=white] (i') at (\x-2,\y-2) {$8$};
        \node[fill=white] (j') at (\x-3,\y-2) {$3$};
        \node[fill=white] (k') at (\x-4,\y-2) {$6$};
        \node[fill=white] (m') at (\x-1,\y-2) {\bf 5}; 
        \node[fill=white] (n') at (\x-0.3,\y-2) {\bf 2};
        \node[fill=white] (l') at (\x+0.3,\y-2) {\bf 9};

        \draw (a')--(b'); 
        \draw (a')--(c');
        \draw (a')--(d');
        \draw (d')--(i');
        \draw (d')--(j');
        \draw (d')--(k');
        \draw (c')--(l');
        \draw (c')--(m');
        \draw (c')--(n');
        \draw (c')--(h');
        \draw (b')--(f');
        \draw (b')--(e');

\end{tikzpicture}
        \end{center}
        \caption{At the top left: a PQ-tree $T$ whose subtree $T_x$
        is an $X$-tree for $X = \{2,5,7,9\}$. At the top right:
        the tree $F(T,X) = \FLAT(T,X)$. At the bottom right: a tree $F(T,X)''$
        that is equivalent to $F(T,X)$. In the middle: the trees $T'$ (on 
        the left) and $F(T,X)'$ (on the right). Compare with the proof of 
        Lemma~\ref{lem:flatten}. The permutation $\pi'$ in the proof
        is the inverse of $\FRONT(F(T,X)'') = (6,3,8,5,2,9,7,4,1)$. 
        The permutation $\pi$ in the proof is the inverse 
        of $(6,3,8,2,7,5,9,4,1)$. \label{fig:flatten} }
\end{figure}
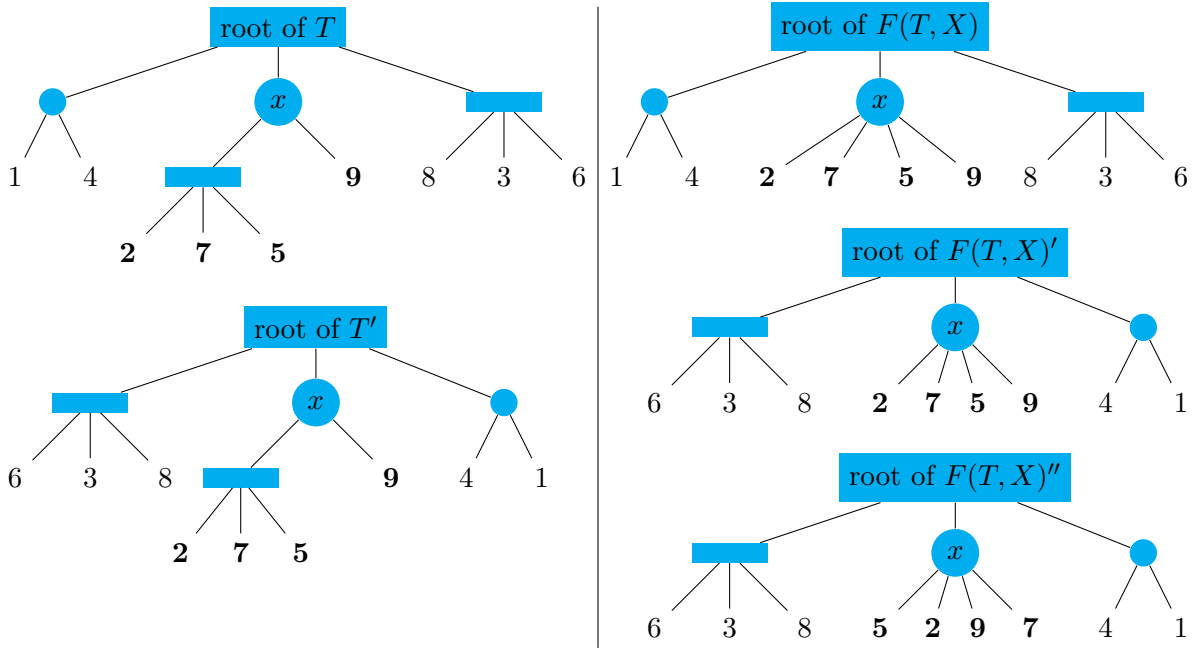

\noindent
The above Claim 1 is implied by the following result:

\begin{corollary} \label{cor:flatten}
For $j = 1,\ldots,m$, the following holds.
If $T_{j-1}$ represents $\Pi_{j-1}$, then
\[
T'_j := \CONS^{-1}(\FLAT(T_{j-1},X_{j-1})) 
\]
represents $\Pi'_j$. Moreover, given a permutation $\pi_j \in \Pi'_j$,
one can compute in time $O(n)$ permutations $\pi_{j-1} \in \Pi_{j-1}$
and $\sigma_{j-1} \in \cS_{X_{j-1}\pi_{j-1}}$ such 
that $\pi_j = \pi_{j-1}\sigma_{j-1}$.
\end{corollary}

\begin{proof}
Apply Lemma~\ref{lem:flatten} for $X = X_{j-1}$ and $T = T_{j-1}$.
Then $\CONS^{-1}(T) = \Pi_{j-1}$ and the right-hand side 
of~(\ref{eq:X-subset}) equals $\Pi'_j$. Since $X_{j-1}$ 
and $T_{j-1}$ satisfy~(\ref{eq:X-subset}), 
Lemma~\ref{lem:flatten} applies. Hence $\FLAT(T_{j-1},X_{j-1})$ 
represents $\Pi'_j$. Moreover, again by Lemma~\ref{lem:flatten}, 
we get the following: given a permutation $\pi' \in \Pi'_j$,
one can compute in time $O(n)$ permutations $\pi \in \Pi_{j-1}$
and $\sigma \in \cS_{X_{j-1}\pi}$ such that $\pi' = \pi\sigma$.
\end{proof}

\noindent
Note that, by Corollary~\ref{cor:flatten}, we have overcome
Obstacle 2 and also Obstacle 3.

\subsection{Putting Everything Together} \label{subsec:main-result}

We are now prepared for the main result in this paper:

\begin{theorem} \label{th:WC1P}
The following can be done in time $O(mn)$:
\begin{enumerate}
\item
Given an instance $I = (n,m,X_1,\ldots,X_m,\tau)$ of WPPSG,
either report that $I$ is not nice or, if $I$ is nice,
compute PQ-trees $T_0,T'_1,T_1,\ldots,T'_m,T_m$ 
such that $T'_j$ represents $\Pi'_j$ and $T_j$ 
\hbox{represents $\Pi_j$}.
\item
Given $I$ and $T_1,T'_2,T_2,\ldots,T'_m,T_m$, compute the 
description $(\pi_1,\sigma_1,\ldots,\sigma_{m-1})$ of an
ascending chain that starts at $I$ and ends at some instance 
of WPPSG$_0$.
\item
Given $I$ and $(\pi_1,\sigma_1,\ldots,\sigma_{m-1})$, compute
the instance $I_m$ that is given by
\begin{equation} \label{eq:instance-recursion}
I_1 = I^{\pi_1}, I_2 = I^1[1,\sigma_1], \ldots, I_m = I_m[m,\sigma_m]
\enspace .
\end{equation}
\item
Decide whether $I_m$ is a YES-instance.
\end{enumerate}
\end{theorem}

\begin{proof}
\begin{enumerate}
\item
Represent the set $\Pi_0 = \cS_n$ by the universal PQ-tree $T_0$.
For $j = 1,\ldots,m$, do the following: 
\begin{description}
\item[Expansion Step:] 
Given $T_{j-1}$, set $T'_j = \FLAT(T_{j-1},X_{j-1})$.
\item[Reduction Step:] 
The method of Booth and Lueker can be used to report that $\Pi_j = \eset$ 
and, if $\Pi_j \neq \eset$, to compute a PQ-tree $T_j$ that represents $\Pi_j$.
\end{description}
The correctness of the above procedure follows inductively:
\begin{itemize}
\item 
If $T_{j-1}$ represents $\Pi_{j-1}$ then, according to
Corollary~\ref{cor:flatten}, $\Pi'_j$ is represented by 
the PQ-tree $T'_j  = \FLAT(T_{j-1},X_{j-1})$.
\item
If $T'_j$ represents $\Pi'_j$, then the method of Booth and Leuker
computes a PQ-tree $T_j$ that represents $\Pi_j$ unless $\Pi_j = \eset$.
\item
If $\Pi_j = \eset$, then $\Pi_m = \eset$. According to
Theorem~\ref{th:reduce-expand}, the latter can happen only
if $I$ is not nice. Therefore $\Pi_j \neq \eset$ provided 
that $I$ is nice.
\end{itemize}
We move on to time analysis.
The construction of $T_0$ takes time $O(n)$. According to
Corollary~\ref{cor:flatten}, the expansion step takes time $O(n)$.
According to~\cite{BL1976}, the reduction step takes time $O(n)$.
Since each of these steps is executed $m$-times, the overall
time bound is $O(mn)$.
\item
\begin{enumerate}
\item
Set $\pi_m := \FRONT(T_m)$. 
\item
For $j = m$ downto $1$, compute 
$\pi_{j-1} \in \Pi_{j-1}$ and $\sigma_{j-1} \in \cS_{X_{j-1}\pi_{j-1}}$
such that $\pi_j = \pi_{j-1}\sigma_{j-1}$.
\item
Return $(\pi_1,\sigma_1,\ldots,\sigma_m)$.
\end{enumerate}
As shown in the proof of Theorem~\ref{th:reduce-expand}, 
the returned tuple, $(\pi_1,\sigma_1,\ldots,\sigma_m)$,
is a description of an ascending chain that starts in $I$
and ends in an instance of WPPSG$_0$. Thus the above procedure
is correct. The run-time is dominated by the loop in Step 2(b).
According to Corollary~\ref{cor:flatten}, each run through the loop
takes time $O(n)$. Since the loop is executed $m$-times,
we obtain the overall time bound $O(mn)$.
\item
Expanding the recursion in~(\ref{eq:instance-recursion}),
we obtain 
\[
I_m = (n,m\ ,\ X_1\pi_1,\ldots,X_m\pi_m\ ,\ \pi_1^{-1}\tau\pi_m)
\]
where $\pi_j = \pi_1\sigma_1\ldots\sigma_{j-1}$. Given $\pi_{j-1}$, 
the permutation $\pi_j = \pi_{j-1}\sigma_{j-1}$ can be computed
in time $O(n)$. Thus the instance $I_m$ can be computed in time $O(mn)$.
\item
It is easy to see that the sorting strategy applied to $\tau_0 \in \cS_n$
and $X_1,\ldots,X_m \seq [n]$ can be implemented in time $O(mn)$.\footnote{The
time complexity is dominated by $m$-times running BUCKET-SORT for sorting
an unsorted sequence consisting of less than $n$ numbers in the range 
from $1$ to $n$.} Since $I_m$ is an instance of WPPSG$_0$,
this strategy can be used to decide in time $O(mn)$ whether $I_m$ 
is a YES-instance.
\end{enumerate}
\end{proof}

\begin{corollary} \label{cor:WC1P}
The problem WPPSG$_2$ can be solved in time $O(mn)$.
\end{corollary}

\begin{proof}
We use the notation from Theorem~\ref{th:WC1P}. If the given instance $I$
is not nice then, as described in the proof of Theorem~\ref{th:WC1P}, this 
will be revealed during the execution of the REDUCE-EXPAND procedure.
We may assume therefore that the given instance $I$ is nice. According to 
Theorem~\ref{th:WC1P}, we can compute the instance $I_m$ and test whether 
it is a YES-instance in time $O(mn)$. By virtue of 
Lemma~\ref{lem:transformation}, $I_m$ is a YES-instance iff $I$ is 
a YES-instance. Thus the problem WPPSG$_2$ can be solved in time $O(mn)$.
\end{proof}

\begin{corollary}
It can be checked in time $O(mn)$ whether a matrix $A \in \{0,1\}^{n \times m}$
has the weak consecutive-ones property.
\end{corollary}

\appendix

\section{Proof of Lemma~\ref{lem:consistency}} \label{sec1:app}

Let $\pi = (1\pi,\ldots,n\pi)$ be the given permutation 
on $[n]$ and let $T$ be the given PQ-tree over $[n]$. 
For each node $x$ \hbox{in $T$}, let $P_{min}(x|\pi)$ 
be the smallest $i \in [n]$ such that $i\pi$ is a leaf in $T_x$. 
The quantity $P_{max}(x|\pi)$ is defined analogously. Let $x_0$ 
denote the root of $T$. It is obvious 
that $P_{min}(x_0|\pi) = 1$ and $P_{max}(x_0|\pi) = n$.
The quantities $P_{min}(x|\pi)$ and $P_{max}(x|\pi)$, with $x$ 
ranging over all nodes of $T$, can be computed in linear 
time in a bottom-up fashion:
\begin{enumerate}
\item
If $x$ is a leaf in $T$, say $x = i$, then we 
set $P_{min}(x|\pi) := P_{max}(x|\pi) := i\pi^{-1}$.
\item
If $x$ is an internal node in $T$ with children $x_1,\ldots,x_r$,
then we set 
\[
P_{min}(x|\pi) := \min_{j=1,\ldots,r}P_{min}(x_j|\pi)\ \mbox{ and }\  
P_{max}(x|\pi) := \max_{j=1,\ldots,r}P_{min}(x_j|\pi) \enspace .
\]
\end{enumerate}
We assume from now that the $P_{min}$- and the $P_{max}$-parameters 
are at our disposal. \\
To solve the consistency problem for the input instance $(T,\pi)$, 
we use an algorithm $A$ that visits the internal nodes of $T$ in 
preorder and, for each currently visited node $x$, does the following:
\begin{enumerate}
\item
Let $x_1,\ldots,x_r$ (in this order) be the children of $x$ in $T$.
\item
If $x$ is a P-node, then let $x'_1,\ldots,x'_r$ be the permutation 
of $x_1,\ldots,x_r$ such that $P_{min}(x'_1|\pi) <\ldots< P_{min}(x'_r|\pi)$.
For reasons of efficiency, we apply BUCKET-SORT here. \\
If $x$ is a Q-node and $P_{min}(x_1|\pi) < P_{min}(x_2|\pi)$, then 
set $x'_j = x_j$ for $j = 1,\ldots,r$. \\
If $x$ is a Q-node and $P_{min}(x_1|\pi) > P_{min}(x_2|\pi)$,
then set $x'_j = x_{r-j+1}$ for $j = 1,\ldots,r$.
\item
Check whether the condition
\begin{equation} \label{eq:consistency}
\forall j=1,\ldots,r-1: P_{max}(x'_j|\pi) < P_{min}(x'_{j+1},\pi)
\end{equation}
is satisfied. If the answer is ``No'', then return 
``error'' and stop. If the answer is ``Yes'', then  
modify $T$ by making $x'_1,\ldots,x'_r$ the new ordering of 
the children of $x$.\footnote{If CPP should not destroy the 
original tree $T$, then CPP could create a copy of $T$ in 
the beginning and make all modifications in this copy.}
\end{enumerate}
When the traversal is finished, $A$ checks whether $\FRONT(T) = \pi$.
If the answer is ``No'', then $A$ reports ``error'' and stops.
If the answer is ``Yes'', then $A$ returns $T$ and stops.\footnote{
This final check is not really necessary but it later simplifies 
the proof of correctness.}
Before proving that $A$ solves the consistency problem, let's
argue that $A$ runs in time $O(n)$. It is rather obvious that,
for each fixed internal node $x$ with $r$ children, Steps 1--3 
can be executed in time $O(r)$. We therefore obtain a time bound
for algorithm $A$ that is proportional to the total number of 
children of internal nodes in $T$ and therefore proportional 
to the number of edges in $T$. Hence $A$ has a time bound $O(n)$. \\
Let's denote by $T$ the original PQ-tree that was part of the input.
During the preorder traversal of $T$, $T$ is modified many times
(one possible modification at each internal node). We denote by $T'$ 
the PQ-tree resulting from these modifications when the traversal is 
finished. \hbox{Clearly $T'$} is equivalent to $T$ because only permitted
modifications were used for the re\-or\-ga\-ni\-za\-tion \hbox{of $T$}.
The proof of the lemma can now be completed by showing
that the frontier of $T'$ equals $\pi$ provided that $\pi \in \CONS(T)$.
If $\pi \notin \CONS(T)$, there is nothing to show. We may therefore
suppose that $\pi \in \CONS(T)$. It follows that there exists
a PQ-tree $T''$ which is equivalent to $T$ and has frontier $\pi$.
Since $T''$ is equivalent to $T$, $T''$ can be obtained from $T$
by visiting the internal nodes of $T$ in preorder and, for each
currently visited node $x$, making a permitted modification
at $x$. This is just like \hbox{algorithm $A$} has transformed $T$
into $T'$. It suffices to show that $T' = T''$. The main
argument for the coincidence of $T'$ and $T''$ is that every
modification that is chosen by $A$ is enforced if we want to 
preserve the chance of having a PQ-tree with frontier $\pi$ at
the end of the preorder traversal. Let's flesh out this
central argument by making the following observations:
\begin{itemize}
\item
The modifications chosen at the internal nodes in $T_x$ determine 
the internal order of the leaves of $T_x$. The modifications chosen
at nodes outside $T_x$ have no influence on this internal order.
\item
The modification chosen at the internal node $x$ determines the
internal order of the leaves in $T_x$ partially: if $x$
has $r$ children $x_1,\ldots,x_r$ (ordered from left to
right) and the modification is to make $x'_1,\ldots,x'_r$ 
the new order of $x_1,\ldots,x_r$, then it is determined
that each leaf of $T_{x'_1}$ is left of each leaf in $T_{x'_2}$,
each leaf of $T_{x'_2}$ is left of each leaf in $T_{x'_3}$, 
and so on. Note that Condition~(\ref{eq:consistency})
controls whether this matches with the left-to-right order
of the corresponding components of $\pi$.
\item
The choice of $x'_1,\ldots,x'_r$ in Step 2 of algorithm $A$ is
enforced: for any other order of $x_1,\ldots,x_r$, there is
an immediate mismatch with $\pi$.
\item
Even the order $x'_1,\ldots,x'_r$ chosen in Step 2 may produce 
a mismatch with $\pi$. But algorithm $A$ takes care of this 
possibility in Step 3 by checking Condition~(\ref{eq:consistency}). 
\end{itemize}
From this discussion and our assumption that $\pi\in\CONS(T)$
is witnessed by $T''$, it follows that $T' = T''$ and 
therefore $\FRONT(T') = \pi$. \hfill \qed

\section{Proof of Lemma~\ref{lem:consecutive-integers}}
\label{sec2:app}

Fig.~\ref{fig:lem-X-tree} and~\ref{fig:proof-X-tree}
serve as an illustration for the following proof.

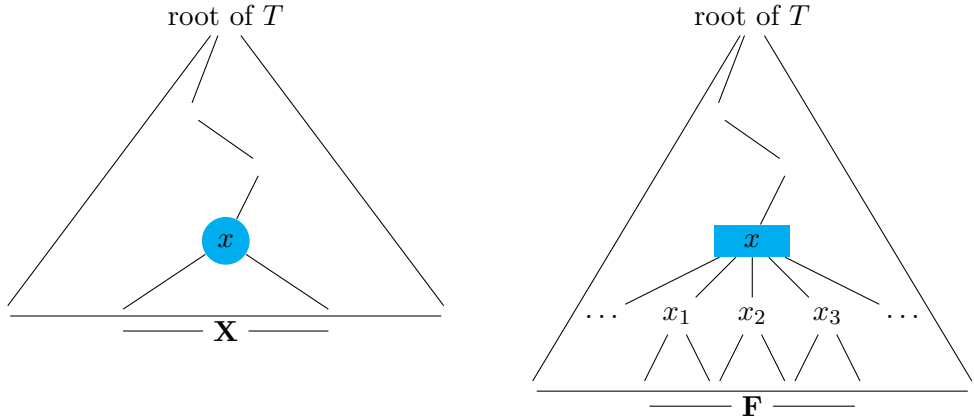
\begin{figure}[hbt] 
        \begin{center}
\begin{tikzpicture}



        \def\x{-3} \def\y{0}

        \node[fill=white] (a) at (\x,\y) {root of $T$};
        \node[fill=white] (a1) at (\x-0.5,\y-1.3) {};
        \node[fill=white] (a2) at (\x+0.5,\y-2) {};
        \node[fill=cyan,circle] (x) at (\x,\y-3) {$x$};
        \node[fill=white] (b) at (\x-3,\y-4) {};
        \node[fill=white] (c) at (\x+3,\y-4) {};
        \node[fill=white] (z1) at (\x-1.5,\y-4) {}; 
        \node[fill=white] (z2) at (\x+1.5,\y-4) {}; 
        \node[fill=white] (z1') at (\x-1.5,\y-4.2) {}; 
        \node[fill=white] (z2') at (\x+1.5,\y-4.2) {}; 
        \node[fill=white] (z) at (\x,\y-4.2) {\bf X};

        \draw (a)--(a1);
        \draw (a1)--(a2);
        \draw (a2)--(x);
        \draw (a)--(b);
        \draw (a)--(c);
        \draw (b)--(c);
        \draw (x)--(z1);
        \draw (x)--(z2);
        \draw (z1')--(z);
        \draw (z2')--(z);

        \def\x{4} \def\y{0}

        \node[fill=white] (a) at (\x,\y) {root of $T$};
        \node[fill=white] (a1) at (\x-0.5,\y-1.3) {};
        \node[fill=white] (a2) at (\x+0.5,\y-2) {};
        \node[fill=cyan,rectangle,minimum width = 1cm] (x) at (\x,\y-3) {$x$};
        \node[fill=white] (x0) at (\x-2,\y-4) {$\ldots$};
        \node[fill=white] (x1) at (\x-1,\y-4) {$x_1$};
        \node[fill=white] (x2) at (\x,\y-4) {$x_2$};
        \node[fill=white] (x3) at (\x+1,\y-4) {$x_3$};
        \node[fill=white] (x4) at (\x+2,\y-4) {$\ldots$};
        \node[fill=white] (b) at (\x-3,\y-5) {};
        \node[fill=white] (c) at (\x+3,\y-5) {};
        \node[fill=white] (z1) at (\x-1.5,\y-5) {}; 
        \node[fill=white] (z2) at (\x+1.5,\y-5) {}; 
        \node[fill=white] (z1') at (\x-1.5,\y-5.2) {}; 
        \node[fill=white] (z2') at (\x+1.5,\y-5.2) {}; 
        \node[fill=white] (z) at (\x,\y-5.2) {\bf F}; 
        \node[fill=white] (z3) at (\x-0.5,\y-5) {}; 
        \node[fill=white] (z4) at (\x+0.5,\y-5) {}; 

        \draw (a)--(a1);
        \draw (a1)--(a2);
        \draw (a2)--(x);
        \draw (a)--(b);
        \draw (a)--(c);
        \draw (b)--(c);
        \draw (x)--(x0);
        \draw (x)--(x1);
        \draw (x)--(x2);
        \draw (x)--(x3);
        \draw (x)--(x4);
        \draw (z1')--(z);
        \draw (z2')--(z);
        \draw (x1)--(z1);
        \draw (x1)--(z3);
        \draw (x2)--(z3);
        \draw (x2)--(z4);
        \draw (x3)--(z4);
        \draw (x3)--(z2);

\end{tikzpicture}
        \end{center}
        \caption{The subtree $T_x$ on the left is an $X$-tree
          that is rooted at a P-node $x$. The subtree $T_x$ on 
          the right contains an $X$-forest and is rooted at
          a Q-node $x$.  \label{fig:lem-X-tree} }
\end{figure}

\begin{figure}[hbt] 
        \begin{center}
\begin{tikzpicture}



        \draw (0,0) -- (-5,-4);
        \draw (0,0) -- (5,-4);
        \draw (-5,-4) -- (5,-4);
        \draw (0,0) -- (-2,-2);  
        \draw (0,0) -- (0,-2);  
        \draw (0,0) -- (2,-2);  
        \draw (-2,-2) -- (-3,-4);
        \draw (-2,-2) -- (-1,-4);
        \draw (0,-2) -- (-1,-4);
        \draw (0,-2) -- (1,-4);
        \draw (2,-2) -- (1,-4);
        \draw (2,-2) -- (3,-4);
        \draw (2,-2) -- (1.7,-3);
        \draw (2,-2) -- (2.3,-3);
        \draw (1.7,-3) -- (1.4,-4);
        \draw (1.7,-3) -- (2,-4);
        \draw (2.3,-3) -- (2,-4);
        \draw (2.3,-3) -- (2.6,-4);
        
        \node (x) at (0.3,0) {x};
        \node (y') at (-1.7,-2) {$y'$};
        \node (y'') at (0.3,-2) {$y''$};
        \node (y) at (1.7,-2) {$y$};
        \node (z') at (-2,-4.3) {$z'$};
        \node (L') at (-2,-4) {$|$};
        \node (z'') at (0,-4.3) {$z''$};
        \node (L'') at (0,-4) {$|$};
        \node (z+) at (1.7,-4.3) {$z^+$};
        \node (L+) at (1.7,-4) {$|$};
        \node (z-) at (2.3,-4.3) {$z^-$};
        \node (L-) at (2.3,-4) {$|$};
        \node (usw) at (2,-3) {$\ldots$};


\end{tikzpicture}
        \end{center}
        \caption{Illustration of the proof of Lemma~\ref{lem:consecutive-integers}:
        reversing the order of the children of the node $y$ has the effect that
        the leaves $z',z'',z^+ \in X$ and $z^- \notin X$ occur in the
        left-right order $z',z'',z^-,z^+$ so the the $X$-leaves do not
        occur consecutively. \label{fig:proof-X-tree} }
\end{figure}
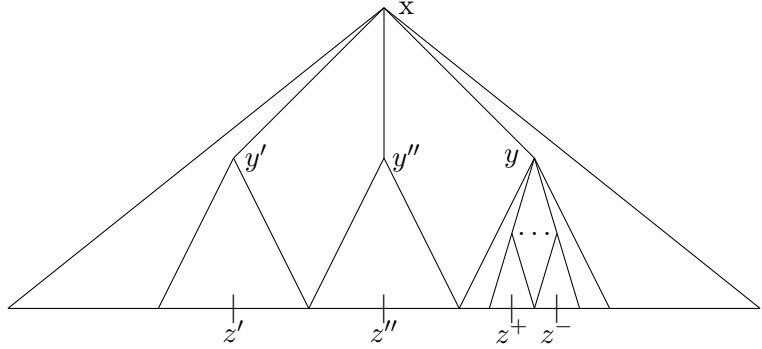

The choice of $x$ implies that $x$ must have two children,
say $y'$ and $y''$, such that each of the subtrees $T_{y'}$ and $T_{y''}$
has a leaf from $X$, say $T_{y'}$ contains the leaf $z' \in X$
and $T_{y''}$ contains the leaf $z'' \in X$. 
\begin{description}
\item[Claim:] 
For each child $y$ of $x$, either all leaves in $T_{y}$ are members of $X$
or none of them belongs \hbox{to $X$}. 
\item[Proof of the Claim:]
Assume for the sake of contradiction that there exists a child $y$ of $x$ 
such that $T_y$ contains a leaf from $X$ and another leaf outside $X$, 
say the leaves $z^+ \in X$  and $z^- \notin X$. The node $y$ may coincide 
with one of the nodes $y'$ and $y''$ but not with both. For reasons of 
symmetry, we may assume that $y \neq y'$. We discuss how the 
leaves in the set $\{z',z^+,z^-\}$ can possibly be ordered from left 
to right:
\begin{itemize}
\item 
Since $T$ satisfies~(\ref{eq:consecutive-integers}), the leaf $z^-$
cannot be between $z'$ and $z^+$.
\item
Since $z^+$ and $z^-$ are leaves in $T_y$ and $z'$ is a leaf outside $T_y$,
it follows that $z'$ cannot be between $z^+$ and $z^-$. 
\item
It follows that, in the ordering of leaves from left to right, the
leaf $z^+$ must be between $z'$ and $z^-$. This leaves the orderings 
$z'z^+z^-$ and $z^-z^+z'$ as the only remaining possibilities.
\end{itemize}
Let $z$ be the youngest common ancestor of $z^+$
and $z^-$ in $T_y$. Let $z_1$ be the child of $z$ with the leaf $z^+$
in $T_{z_1}$ and let $z_2$ be the child of $z$ with the leaf $z^-$ 
in $T_{z_2}$. By reversing the order of the children of $z$, we obtain 
a PQ-tree $T'$ that is equivalent to $T$ but has the property that 
the leaves $z'$, $z^+$ and $z^-$ occur in $T'$ in the 
order $z'z^-z^+$ or in the order $z^+z^-z'$.\footnote{A reversal
of the order of the children of $z$ is possible regardless of whether $z$
is a P- or a Q-node.} We may conclude that there exists 
a permutation $\pi \in \CONS(T)$ such that $X\pi$ is not a set 
of consecutive integers. We arrived at a contradiction, wich proves 
the claim.
\end{description}
From the claim, the assertion of the lemma is obvious.


\end{document}